\documentclass[10pt, twocolumn, comsoc]{IEEEtran}

\usepackage{graphicx,epsfig}
\usepackage[noadjust]{cite}
\usepackage{amsfonts,helvet}
\usepackage{fancyhdr}
\usepackage{threeparttable}
\usepackage{epsf}
\usepackage{amsthm} 
\usepackage{booktabs} 

\usepackage{amsmath}
    \DeclareMathOperator{\tr}{tr}
    \DeclareMathOperator{\diag}{diag}
    \DeclareMathOperator{\rank}{rank}
    \DeclareMathOperator*{\argmax}{arg\,max}
    
\usepackage{algorithmicx}

\usepackage{siunitx}
\usepackage{amssymb}
\usepackage{dsfont}
\usepackage{subfigure}
\usepackage[svgnames]{xcolor}
\usepackage[linesnumbered,ruled]{algorithm2e}
\usepackage{enumerate}
\usepackage{cancel}
\usepackage{comment}
\usepackage{blkarray}
\usepackage[normalem]{ulem}
\usepackage{newtxtext,newtxmath}

\newtheorem{theorem}{Theorem}

\newtheorem{lemma}{Lemma}

\newtheorem{proposition}{Proposition}
\newtheorem{remark}{Remark}

\SetKwInput{KwInput}{Initialize}
\SetKwInput{KwOutput}{Output}

\usepackage{eucal}

\usepackage[colorlinks=true, linkcolor=blue]{hyperref}
\begin{document}

\title{
The MIMO-ME-MS Channel: \\Analysis and Algorithm for Secure MIMO Integrated Sensing and Communications
}

\author{Seongkyu~Jung, Namyoon~Lee, and Jeonghun~Park

\thanks{ 
This work was supported by Institute of Information \& communications Technology Planning \& Evaluation (IITP) grant funded by the Korea government (MSIT) (No. RS-2024-00397216, Development of the Upper-mid Band Extreme massive MIMO (E-MIMO)), in part by the IITP grant funded by the Korea government (MSIT) (No. RS-2024-00395824, Development of Cloud virtualized RAN (vRAN) system supporting upper-midband), and in part by the IITP under the 6G Cloud Research and Education Open Hub (IITP-2025-RS-2024-00428780) grant funded by the Korea government (MSIT).}
\thanks{
S. Jung and J. Park are with the School of Electrical and Electronic Engineering, Yonsei University, Seoul 03722, South Korea (e-mail: wjdtjd963@yonsei.ac.kr; jhpark@yonsei.ac.kr). N. Lee is with the Department of Electrical Engineering, POSTECH, Pohang 37673, South Korea (e-mail: nylee@postech.ac.kr).}
} 

\maketitle

\begin{abstract}
This paper addresses precoder design for secure multiple-input multiple-output (MIMO) integrated sensing and communications (ISAC) systems. We introduce a MIMO channel with a multiple-antenna eavesdropper and a multiple-antenna sensing receiver (MIMO-ME-MS) and analyze the fundamental performance limits of this tripartite tradeoff. Using sensing mutual information, we formulate the precoder design as a nonconvex weighted sum rate maximization problem. A high signal-to-noise ratio analysis based on a subspace decomposition characterizes the maximum weighted degrees of freedom. This analysis reveals the structure of a quasi-optimal precoder that must span the ``useful subspace'' and demonstrates the inadequacy of extending known schemes from simpler wiretap or ISAC channels. To solve this nonconvex problem, we develop a practical two-stage iterative algorithm that alternates between a sequential basis construction stage and a power allocation stage that solves the resulting difference-of-convex program. We demonstrate that the proposed method captures the desirable precoder structure identified in our analysis and achieves substantial performance gains in the MIMO-ME-MS channel.
\end{abstract}

\section{Introduction}
A leading trend in next-generation wireless systems is the integration of sensing functionality into conventional communication infrastructure, an approach commonly termed integrated sensing and communications (ISAC). 
In particular, multiple-input multiple-output (MIMO) technology enables ISAC systems to exploit their abundant spatial degrees of freedom (DoF) to serve communication and sensing functions simultaneously.
By jointly harnessing the broadcast nature of the wireless channel and MIMO's beamforming capability, the transmitter can judiciously reuse a single waveform for both data transmission and target probing. This joint use can achieve synergistic gains in communication mutual information (MI) and sensing metrics such as the Cram\'er--Rao lower bound (CRLB) \cite{CaireFTISAC2023,fanDRTisit23}.

The broadcast nature of wireless propagation, however, also exposes transmissions to eavesdropping threats. 
A passive eavesdropper, possibly equipped with multiple antennas, can capture any signal transmitted within the transmitter's coverage area. 
A standard information-theoretic countermeasure is physical-layer security (PLS). 
Wyner's seminal work \cite{wynerWiretap75} showed that a transmitter can deliver a confidential message to a legitimate user at a positive secrecy rate, 
which is characterized by the difference between the MI at the legitimate user and that at the eavesdropper. 
Accordingly, if the legitimate user's channel is stronger than the eavesdropper's, the secrecy rate is strictly positive.
In MIMO settings, the effective channel is shaped by the transmit precoders, making the characterization and optimization of the secrecy rate nontrivial. 
Addressing this, the secrecy capacity of the MIMO wiretap channel was fully characterized in \cite{KhistiMIMOME2010}, which showed that generalized singular value decomposition (GSVD)-based precoding achieves this capacity in the high signal-to-noise ratio (SNR) regime.
Nonetheless, prior work on PLS \cite{wynerWiretap75, KhistiMIMOME2010} considered only two types of receivers, i.e., a legitimate user and an eavesdropper, without accounting for sensing functionality. As a result, the existing framework is insufficient for understanding secure MIMO ISAC systems.

In this paper, we explore the MIMO-ME-MS channel---a MIMO channel with a multiple-antenna eavesdropper and a multiple-antenna sensing receiver---which extends the MIMO-ME channel \cite{KhistiMIMOME2010} by incorporating sensing functionality. Using the concept of sensing mutual information (SMI) \cite{smicommmag23} as the sensing metric, we characterize optimality conditions and propose a quasi-optimal transmission method in the high-SNR regime. 
Additionally, we develop a practical method for precoding basis design and power allocation to support joint communication, sensing, and secrecy. 

\subsection{Related Work}

The literature contains extensive prior work on secure MIMO and MIMO ISAC systems, although these two areas are largely treated separately. 
A key result in secure MIMO communication is provided in \cite{KhistiMIMOME2010}, which characterizes the secrecy capacity of the MIMO wiretap channel and establishes the high-SNR optimality of GSVD-based precoding. Complementing this, the same capacity expression is derived in \cite{hassibiWiretap11} using a different proof technique (a saddle-point characterization in \cite{KhistiMIMOME2010} versus a single convex optimization in \cite{hassibiWiretap11}). 
In \cite{choiColluding}, a scenario with cooperating eavesdroppers is considered, and an iterative precoding algorithm is developed to maximize the sum secrecy rate based on a generalized power-iteration approach \cite{parkGPIRS23}.
In \cite{hiaTcom22, zhang:tvt:19}, hierarchical PLS is introduced, where higher-tier users can decode messages intended for lower-tier users, but not vice versa.
In \cite{salemRSMA23, xiatwc24}, precoder optimization methods are proposed to enhance the secrecy rate by employing rate-splitting multiple access \cite{parkNet}. 
A comprehensive survey of PLS is presented in \cite{wuSurvey18}. 

For MIMO ISAC systems, the fundamental tradeoff between communication MI and the sensing CRLB is characterized in \cite{CaireFTISAC2023, mimoCRBTWC24}.
For MIMO ISAC precoder optimization, semidefinite programming (SDP)-based methods are developed in \cite{liuSDPtsp20, liuFanCRLBopt22} to enhance the sensing accuracy while satisfying prescribed signal-to-interference-plus-noise ratio (SINR) constraints.
In \cite{choiGPIisac24}, a generalized power iteration-based precoding method is proposed to maximize the achievable sum rate subject to beampattern mean-squared error (MSE) constraints.
In \cite{kim:twc:26}, considering frequency-division duplexing, a downlink channel-reconstruction technique \cite{Darktwc25} is studied in the ISAC context. 
One key challenge in studying MIMO ISAC lies in the disparity between communication metrics (e.g., MI) and sensing metrics (e.g., CRLB and beampattern MSE). 
To address this, recent work has employed SMI as a sensing metric, inspired by the use of MI as an information-theoretic performance metric for sensing \cite{bell:tit:93, tang:tsp:10, tang:tsp:19}. 
For instance, in \cite{wangUnifiedISACPareto2024}, the relationship between SMI and minimum mean-square error (MMSE) on the ISAC Pareto boundary is investigated, and a weighted MMSE (WMMSE)-based precoding algorithm is developed.
In \cite{shin:arxiv:25}, greedy radio-frequency chain selection methods are devised by leveraging a unified MI-based ISAC performance characterization. 

Recent literature on secure MIMO ISAC has extensively investigated the joint design of communication and sensing signals to balance secrecy requirements with sensing performance. A prevalent approach involves the simultaneous optimization of information and artificial noise signals to maximize the secrecy rate while satisfying sensing constraints such as beampattern error, radar SINR, or CRLB \cite{su:twc:21, jia:tvt:24, ren:tcom:23}.
To address practical challenges, studies have also accounted for eavesdropper channel state information (CSI) uncertainty by employing robust optimization to guarantee worst-case secrecy performance \cite{ren:tcom:23, jia:tvt:24}, or have leveraged sensing functionality itself to localize unknown eavesdroppers for enhanced security \cite{su:twc:24}. 
Furthermore, data-driven frameworks based on deep learning have emerged as an alternative approach to precoder design \cite{li:commlett:24}.

The technical core of these existing designs relies primarily on sophisticated numerical optimization frameworks. 
Specifically, nonconvex secrecy-oriented problems are typically tackled by
reformulating them into tractable forms via techniques such as semidefinite relaxation (SDR) of rank-one constraints \cite{chu:tvt:23, li:tvt:25}, successive convex approximation (SCA) based on first-order Taylor expansions \cite{he:twc:24, li:tvt:25}, and the S-procedure for handling bounded CSI errors \cite{ren:tcom:23, jia:tvt:24}.
While these methods can produce high-performance beamforming solutions via interior-point methods or standard solvers like CVX, they are limited in revealing and incorporating the structural insights of the optimal precoder design.

Consequently, these approaches provide limited insight into a fundamental question: how an optimal secure ISAC precoder should be structured relative to the tripartite relationship among communication, sensing, and eavesdropper's channel subspaces. Moreover, the fundamental limits of this tradeoff in terms of DoF remain largely uncharacterized. In particular, this lack of structural understanding obscures whether and how known optimal designs for simpler channels (e.g., GSVD-based precoding for the MIMO wiretap channel \cite{KhistiMIMOME2010}) can be extended to secure MIMO ISAC systems.

To bridge this fundamental gap, we provide a rigorous analysis of the MIMO-ME-MS channel that reveals its intrinsic performance limits and the structure of a quasi-optimal precoder. Leveraging this structural insight, we propose a practical precoding algorithm for the MIMO-ME-MS channel. 

\subsection{Contributions}
The main contributions of this paper are summarized as follows: 
\begin{itemize}
    \item \textbf{Unified analytical framework for the MIMO-ME-MS channel:} 
    By adopting SMI as an information-theoretic sensing metric, we introduce the MIMO-ME-MS channel as a tractable model for systems with simultaneous secure communication and sensing requirements. 
    On this basis, we formulate a unified precoder design problem that captures the fundamental tradeoffs among secrecy, communication, and sensing performance within a single weighted sum rate maximization framework.

   \item \textbf{Characterization of the DoF-optimal precoder structure:} 
    We present a rigorous high-SNR analysis to characterize the structure of a quasi-optimal precoder for the MIMO-ME-MS channel. By decomposing the transmit space into eight subspaces, we derive the maximum achievable weighted DoF. 
    This analysis reveals that a DoF-optimal precoder must exclusively span a ``useful subspace,'' whose composition depends on the system weights, and shows that naive extensions of known precoding methods are strictly suboptimal.

    \item \textbf{Two-stage precoding design:} 
    Building on these insights, we propose a practical two-stage iterative algorithm to solve the precoder design problem. The algorithm alternates between (i) sequential basis construction, which maximizes the marginal rate gain at each step, and (ii) power allocation, which solves the resulting difference-of-convex (DC) program. 
    We show that the algorithm's behavior aligns with the asymptotically optimal structure identified by our analysis.
    This theoretical soundness is further validated through numerical simulations, which demonstrate superior performance over baseline schemes across all SNR regimes.
\end{itemize}

The rest of this paper is organized as follows. Section~\ref{sec:system_model} presents the system model for the MIMO-ME-MS channel. Section~\ref{sec:analysis} conducts a theoretical analysis of the problem and characterizes the optimal precoder structure in the high-SNR regime. Building on these theoretical foundations, Section~\ref{sec:algorithm} details the proposed two-stage iterative precoding algorithm and establishes its asymptotic optimality. Section~\ref{sec:simulations} validates the performance of the proposed design through numerical simulations, and Section~\ref{sec:conclusion} concludes the paper.

\textbf{Notation}: Scalars are denoted by italic letters (e.g., $a$), vectors by bold lowercase letters (e.g., $\mathbf{x}$), matrices by bold uppercase letters (e.g., $\mathbf{X}$), and subspaces by calligraphic letters (e.g., $\CMcal{S}$). 
The superscripts $(\cdot)^T$, $(\cdot)^H$, and $(\cdot)^{\dagger}$ denote the transpose, Hermitian (conjugate) transpose, and pseudoinverse, respectively. 
The Euclidean norm of a vector $\mathbf{x}$ is denoted by $\|\mathbf{x}\|_2$.
The trace and rank of a matrix $\mathbf{X}$ are denoted by $\tr(\mathbf{X})$ and $\rank(\mathbf{X})$, respectively. 
The identity matrix is denoted by $\mathbf{I}$, and $\diag(x_1,\dots,x_n)$ 
denotes the diagonal matrix with diagonal entries $x_1,\dots,x_n$; 
equivalently, $\diag(\mathbf{x})$ denotes the diagonal matrix formed 
from the entries of a vector $\mathbf{x}$.
The set of $m \times n$ complex matrices is denoted by $\mathbb{C}^{m\times n}$, and $\mathbb{R}_+$ denotes the set of nonnegative real numbers.
The column space and null space of a matrix $\mathbf{X}$ are denoted by $\CMcal{C}(\mathbf{X})$ and $\CMcal{N}(\mathbf{X})$, respectively. 
The direct sum of subspaces is denoted by $\oplus$, and $\bigoplus_{i=1}^{n}\CMcal{S}_i$ denotes the direct sum of subspaces $\CMcal{S}_1, \dots, \CMcal{S}_n$.
Moreover, $\mathcal{CN}(\mathbf{0},\mathbf{R})$ denotes the circularly symmetric complex Gaussian distribution with zero mean and covariance matrix $\mathbf{R}$, and $[x]^+$ denotes the positive-part operator defined as $\max\{0,x\}$.
Unless otherwise stated, $\log(\cdot)$ denotes the base-2 logarithm.
We use $O(\cdot)$ and $o(\cdot)$ to denote the big-$O$ and little-$o$ notations, respectively.
Specifically, $f(n) = O(g(n))$ implies $\limsup_{n \to \infty} |f(n)/g(n)| < \infty$, and $f(n) = o(g(n))$ implies $\lim_{n \to \infty} f(n)/g(n) = 0$.

\section{System Model} \label{sec:system_model}
We consider a MIMO-ME-MS channel comprising a transmitter (TX) with $n_t$ antennas that serves a legitimate receiver (RX) with $n_c$ antennas in the presence of a passive eavesdropper with $n_e$ antennas. Simultaneously, the TX uses the same waveform to perform target sensing, and the reflected signals are captured by a sensing receiver with $n_s$ antennas.

\subsection{Transmit Signal Model}
Let $N_s$ denote the number of transmitted data streams, and $\mathbf{S}\in\mathbb{C}^{N_s\times T}$ denote the data-symbol matrix, where $T$ is the blocklength. 
We adopt a block-fading model assuming that the transmission duration $T$ is within the coherent sensing period \cite{CaireFTISAC2023}, which implies that the responses of all involved channels remain constant during each block. The entries of $\mathbf{S}$ are independent and identically distributed (i.i.d.) as $\mathcal{CN}(0,1)$, satisfying $\mathbb{E}[\mathbf{S}\mathbf{S}^H] = T \mathbf{I}$.
The TX employs a linear precoding matrix $\mathbf{F}\in\mathbb{C}^{n_t\times N_s}$ to produce the transmitted signal matrix:
\begin{align}
    \mathbf{X}= \mathbf{F}\,\mathbf{S}\in\mathbb{C}^{n_t\times T}.
\end{align}
The transmission is subject to a total average transmit power constraint $P_{\mathrm{tot}}$, which is expressed as:
\begin{align}
    \frac{1}{T}\mathbb{E}\bigl[\tr(\mathbf{X}\mathbf{X}^H)\bigr]
    = \frac{1}{T}\tr(\mathbf{F}\,\mathbb{E}[\mathbf{S}\mathbf{S}^H]\mathbf{F}^{H}) = \tr(\mathbf{F}\mathbf{F}^{H}) \le P_{\mathrm{tot}}.
\end{align}

\subsection{Communication Model}
Let $\bar{\mathbf{H}}_c\in\mathbb{C}^{n_c\times n_t}$ denote the channel matrix from the TX to the legitimate RX. The received signal at the RX is given by:
\begin{align}
    \mathbf{Y}_c = \bar{\mathbf{H}}_c\,\mathbf{X}+\mathbf{Z}_c
    = \bar{\mathbf{H}}_c\,\mathbf{F}\mathbf{S}+\mathbf{Z}_c,
\end{align}
where $\mathbf{Z}_c\in\mathbb{C}^{n_c\times T}$ is an additive white Gaussian noise (AWGN) matrix whose columns are i.i.d. as $\mathcal{CN}(\mathbf{0},\sigma_c^2\mathbf{I})$. The communication MI, representing the achievable communication rate, is given by:
\begin{align}
    R_c(\mathbf{F})
    =
    \log\det \left(\mathbf{I}
        +\mathbf{F}^H\left(\frac{1}{\sigma_c^2}\bar{\mathbf{H}}_c^H\bar{\mathbf{H}}_c\right)\mathbf{F}\right).
\end{align}
Note that we have normalized the MI by the blocklength $T$ (i.e., bits per channel use). For notational simplicity, we define the effective communication channel as $\mathbf{H}_c \triangleq \frac{1}{\sigma_c}\bar{\mathbf{H}}_c$, which simplifies the MI expression to:
\begin{align}
    R_c(\mathbf{F})
    =
    \log\det\bigl(\mathbf{I}
        +\mathbf{F}^H\mathbf{H}_c^H\mathbf{H}_c\mathbf{F}\bigr). \label{cmi}
\end{align}

\subsection{Secrecy Model}
Similarly, let $\bar{\mathbf{H}}_e\in\mathbb{C}^{n_e\times n_t}$ denote the channel matrix from the TX to the eavesdropper. Following the standard MIMO wiretap channel model \cite{KhistiMIMOME2010}, we assume the TX has perfect CSI of the eavesdropper's channel. The received signal at the eavesdropper is given by:
\begin{align}
    \mathbf{Y}_e
    =
    \bar{\mathbf{H}}_e\,\mathbf{X}+\mathbf{Z}_e
    =
    \bar{\mathbf{H}}_e\,\mathbf{F}\mathbf{S}+\mathbf{Z}_e,
\end{align}
where $\mathbf{Z}_e \in \mathbb{C}^{n_e \times T}$ is an AWGN matrix whose columns are i.i.d. as $\mathcal{CN}(\mathbf{0},\sigma_e^2\mathbf{I})$. The MI at the eavesdropper is given by:
\begin{align}
    R_e(\mathbf{F})
    =
    \log\det\left(\mathbf{I}
        +\mathbf{F}^H\left(\frac{1}{\sigma_e^2}\bar{\mathbf{H}}_e^H\bar{\mathbf{H}}_e\right)\mathbf{F}\right).\label{eq:MIeve}
\end{align}
By defining the effective eavesdropper channel $\mathbf{H}_e \triangleq \frac{1}{\sigma_e}\bar{\mathbf{H}}_e$, the MI expression simplifies to:
\begin{align}
    R_e(\mathbf{F})
    =
    \log\det\bigl(\mathbf{I}
        +\mathbf{F}^H\mathbf{H}_e^H\mathbf{H}_e\mathbf{F}\bigr). \label{emi}
\end{align}
For a linear precoder $\mathbf{F}$ and Gaussian signaling, the achievable secrecy rate is lower-bounded by the difference between the MI at the RX and the MI at the eavesdropper \cite{Li:ciss:07}:
\begin{align}
    R_{\mathrm{sec}}(\mathbf{F}) = [R_c(\mathbf{F}) - R_e(\mathbf{F})]^+.
    \label{eq:R_sec}
\end{align}
We focus on operating points for which secure communication is feasible, i.e., $R_c(\mathbf{F}) \ge R_e(\mathbf{F})$.

\subsection{Sensing Model}
The TX also performs target sensing using the same waveform. The random channel $\bar{\mathbf{H}}_s \in \mathbb{C}^{n_s \times n_t}$ represents the round-trip target response that the TX aims to estimate. The received sensing signal is given by:
\begin{align}
    \mathbf{Y}_s = \bar{\mathbf{H}}_s\,\mathbf{X} + \mathbf{Z}_s
    = \bar{\mathbf{H}}_s\,\mathbf{F}\mathbf{S} + \mathbf{Z}_s,
\end{align}
where $\mathbf{Z}_s \in \mathbb{C}^{n_s \times T}$ is an AWGN matrix whose columns are i.i.d. as $\mathcal{CN}(\mathbf{0}, \sigma_s^2\mathbf{I})$. 
We adopt a statistical model for the sensing channel, assuming the rows of $\bar{\mathbf{H}}_s$ are i.i.d. as $\mathcal{CN}(\mathbf{0}, \mathbf{R}_{\bar{\mathbf{H}}_s})$, reflecting an extended target with a rich-scattering response. For a sufficiently large blocklength $T$, the sample covariance of the data symbols approximates its expectation, i.e., $\frac{1}{T} \mathbf{S} \mathbf{S}^H \approx \mathbf{I}$.
We employ SMI as the sensing metric, which measures the information about the random target channel $\bar{\mathbf{H}}_s$ contained in the observations \cite{Yang:taes:07, tang:tsp:19}. Under the Gaussian model, the SMI is expressed as:
\begin{align} 
    R_s(\mathbf{F}) \approx n_s \log \det \left( \mathbf{I} + \mathbf{F}^H \left(\frac{T}{\sigma_s^2} \mathbf{R}_{\bar{{\mathbf{H}}}_s}\right) \mathbf{F} \right).
\end{align}
To unify the problem structure, we define the effective sensing channel as a factor $\mathbf{H}_s$ satisfying
$\frac{T}{\sigma_s^2}\mathbf{R}_{\bar{\mathbf{H}}_s}=\mathbf{H}_s^H\mathbf{H}_s$,
which incorporates the target statistics, processing gain $T$, and sensing noise into a single matrix.
Such a factor always exists since $\mathbf{R}_{\bar{\mathbf{H}}_s}$ is a covariance matrix and hence Hermitian positive semidefinite.
By omitting the scalar factor $n_s$ for simplicity, the SMI is rewritten in a form identical to the MI expressions:
\begin{align} 
    R_s(\mathbf{F}) = \log \det \left( \mathbf{I} + \mathbf{F}^H \mathbf{H}_s^H \mathbf{H}_s \mathbf{F} \right). \label{smi_2}
\end{align}

\begin{remark}[Operational meaning of SMI] \normalfont
Motivated by rate-distortion theory, SMI has recently emerged as a fundamental bridge between information measures and estimation performance. 
Conventionally, this principle underpins MI-based radar waveform design \cite{tang:tsp:19, Yang:taes:07}, which maximizes the MI between the observations and the target response. 
Recently, in the context of ISAC systems, several studies have employed SMI as the sensing metric \cite{wangUnifiedISACPareto2024, smicommmag23, Rethinktvt23}. 
Under a Gaussian linear model, where the received sensing signal depends linearly on the target response, SMI is tightly connected to estimation accuracy; specifically, increasing SMI is equivalent to minimizing the MMSE of the target response \cite{Yang:taes:07}, which can improve the estimation accuracy of spatial parameters such as angles. 
When the Gaussian linear model does not hold, by the data-processing inequality, SMI provides an upper bound on the MI associated with the sensing target \cite{Rethinktvt23}. 
Thus, SMI serves as a valuable performance metric for ISAC systems, acting as a useful surrogate objective that correlates with detection probability \cite{bell:tit:93} and estimation performance \cite{smicommmag23}. 
Another notable advantage is its mathematical compatibility with other MI expressions. For instance, the MI expressions for communication, eavesdropping, and sensing take identical forms, as shown in \eqref{cmi}, \eqref{emi}, and \eqref{smi_2}.
This unified formulation facilitates analysis of the optimal precoder structure for the MIMO-ME-MS channel.
We elaborate on this in the next section.
\end{remark}

\subsection{Problem Formulation}

Our objective is to design the precoder $\mathbf{F}$ that maximizes a weighted sum of the secrecy rate and the SMI. The corresponding objective function is defined as: 
\begin{align}
    R(\mathbf{F}) &= w_c R_{\mathrm{sec}}(\mathbf{F}) + w_s R_{s}(\mathbf{F}),
    \label{eq:R_wmi}
\end{align}
where $w_c, w_s \ge 0$ and $w_c + w_s =1$ are weights that control the tradeoff between secure communication and sensing performance.
The optimization problem can thus be expressed using the unified effective-channel models:

\begin{align}
    \label{eq:prob-log}
    \max_{\mathbf{F}} \quad \hspace{-0.3em} & w_c \hspace{-0.1em} \log\det(\mathbf{I} + \mathbf{F}^H \mathbf{H}_c^H \mathbf{H}_c \mathbf{F}) - w_c \hspace{-0.1em} \log\det(\mathbf{I} + \mathbf{F}^H \mathbf{H}_e^H \mathbf{H}_e \mathbf{F}) \nonumber \\
    & + w_s \hspace{-0.1em} \log\det(\mathbf{I} + \mathbf{F}^H \mathbf{H}_s^H \mathbf{H}_s \mathbf{F}) \\
    \text{s.t.} \quad \hspace{-0.3em} & \tr(\mathbf{F}\mathbf{F}^H) \le P_{\mathrm{tot}}. \nonumber
\end{align}

Any precoder can be decomposed via singular value decomposition (SVD) as $\mathbf{F} = \mathbf{U}\mathbf{\Sigma}\mathbf{V}^H$. Because the objective function contains only terms of the form $\log\det(\mathbf{I}+\mathbf{F}^H\mathbf{A}\mathbf{F})$, the right unitary matrix $\mathbf{V}^H$ does not affect the objective value (unitary-invariance: $\det(\mathbf{I}+ \mathbf{V}\mathbf{X}\mathbf{V}^H)=\det(\mathbf{I}+\mathbf{X})$). Furthermore, the transmit power constraint $\tr(\mathbf{F}\mathbf{F}^H)$ is also independent of $\mathbf{V}$. Thus, without loss of generality, we can restrict $\mathbf{F}$ to the structure $\mathbf{F} = \mathbf{W}\mathbf{P}^{1/2}$, where $\mathbf{W} \in \mathbb{C}^{n_t \times N_s}$ is a semi-unitary matrix ($\mathbf{W}^H\mathbf{W}=\mathbf{I}$) representing the precoding basis, and $\mathbf{P}\in \mathbb{R}_+^{N_s \times N_s}$ is a diagonal matrix denoting the per-stream powers.

Despite this simplification, the optimization problem \eqref{eq:prob-log} remains highly challenging due to multiple sources of nonconvexity. 
Specifically, the joint optimization over the precoding basis $\mathbf{W}$ and the power allocation $\mathbf{P}$ is nonconvex due to their bilinear coupling. 
Furthermore, the semi-unitary constraint on $\mathbf{W}$ restricts the feasible set to a nonconvex manifold. 
The secrecy objective introduces an additional, more severe layer of complexity: the difference-of-log-det structure makes the problem nonconcave even when $\mathbf{W}$ is fixed, precluding direct use of standard convex optimization techniques for the power allocation subproblem. 
These difficulties motivate the analysis of the precoder structure in the next section.

\section{Optimal Precoder Structure Analysis} \label{sec:analysis}
In this section, we analyze the optimal precoder structure and extract insights to guide practical precoder design for the MIMO-ME-MS channel. 
For each effective channel matrix $\mathbf{H}_i$ ($i \in \{c,e,s\}$), let $\CMcal{R}_i$ and $\CMcal{N}_i$ denote the row and null spaces, respectively: 
\begin{align}
    \CMcal{R}_i \triangleq \CMcal{C}(\mathbf{H}^H_i), \quad \CMcal{N}_i \triangleq \CMcal{N}(\mathbf{H}_i), \quad \text{for} \; i \in \{c, s, e\}.
\end{align}
To characterize the system performance in the high-SNR regime, we define the DoF of channel $i$ achieved by a precoder $\mathbf{F}$ as:
\begin{align}
    \label{eq:DoF_def}
    d_i(\mathbf F)\;\triangleq\;
    \lim_{P\to\infty}\frac{R_i(\mathbf F;P)}{\log_2 P}.
\end{align}
Accordingly, we denote by $d(\mathbf{F})$ the weighted DoF corresponding to the objective function in \eqref{eq:prob-log}. 

For a single channel matrix $\mathbf{H}_c$, the transmit space $\mathbb{C}^{n_t}$ can be decomposed into the orthogonal direct sum of its row space $\CMcal{R}_c$ and null space $\CMcal{N}_c$ (i.e., $\CMcal{R}_c \oplus \CMcal{N}_c$ with $\CMcal{R}_c \perp \CMcal{N}_c$). In this conventional single-user MIMO setting, the optimal precoder is constructed via the SVD of ${\mathbf{H}}_c$, allocating transmit power to the dominant eigenmodes within $\CMcal{R}_c$ (e.g., via water-filling), while avoiding the null space $\CMcal{N}_c$. However, in the MIMO-ME-MS case, the TX must simultaneously account for three distinct effective channels: ${\mathbf{H}}_c$, ${\mathbf{H}}_e$, and ${\mathbf{H}}_s$. This coupling renders a straightforward application of SVD-based precoding insufficient, as the optimal strategy requires balancing conflicting objectives across non-orthogonal subspaces. This complexity necessitates a more sophisticated analysis. 

We begin with the MIMO-ME channel, a special case of MIMO-ME-MS obtained by setting $w_s = 0$. The optimal precoder structure for this scenario is known in the high-SNR regime \cite{KhistiMIMOME2010}. Analyzing this regime yields valuable insights into the optimal precoder's structure and serves as a foundation for the more general case.

\subsection{MIMO-ME Channel} \label{sec:mimome}
Focusing on the interaction between the communication and eavesdropper channels, we partition the transmit space $\mathbb{C}^{n_t}$ based on the interplay between their respective row and null spaces. This yields a direct sum decomposition of four subspaces:
\begin{align}
    \mathbb{C}^{n_t} = 
    \underbrace{(\CMcal{V}_n^\perp\cap\CMcal{N}_e)}_{\text{Comm.-private}}
    \oplus
    \underbrace{(\CMcal{V}_n^\perp\cap\CMcal{N}_c)}_{\text{Eve.-private}}
    \oplus
    \underbrace{(\CMcal{R}_c\cap\CMcal{R}_e)}_{\text{common}}
    \oplus
    \underbrace{(\CMcal{N}_c\cap\CMcal{N}_e)}_{\text{total-null}},
    \label{eq:me_spaces}
\end{align} 
where $\CMcal{V}_n=\CMcal{N}_c\cap\CMcal{N}_e$. This decomposition constitutes a special case of the general framework established in Theorem~\ref{thm:mems_decomp} (Section~\ref{sec:mimomems}). Since the rate characteristics differ significantly across these subspaces, the precoder must be carefully structured to exploit their distinct contributions. We analyze each subspace below.

\begin{itemize}
    \item \textbf{Comm.-private subspace ($\CMcal{V}_n^\perp \cap\CMcal{N}_e$):} This subspace is ideal for secure transmission since the signal is nulled at the eavesdropper, i.e., $R_e(\mathbf{F})=0$.
    Thus, the secrecy rate scales logarithmically with transmit power, providing a positive DoF gain. To maximize this gain, it is desirable to allocate a dominant share of the power, on the order of $O(P_{\mathrm{tot}})$, to these directions.

     \item \textbf{Common subspace ($\CMcal{R}_c\cap\CMcal{R}_e$):} 
     This is a contested subspace where both parties (RX and eavesdropper) receive the signal. As the transmit power increases, the MIs at both the RX and the eavesdropper grow logarithmically. 
     Therefore, the secrecy rate converges to a constant gain or loss determined by the channel strength ratio between the RX and the eavesdropper within this subspace, contributing zero DoF. 
     To harvest the positive constant gain (in directions where the communication channel is stronger), only a vanishingly small power allocation, on the order of $o(P_{\mathrm{tot}})$, is sufficient.

    \item \textbf{Eve.-private ($\CMcal{V}_n^\perp \cap \CMcal{N}_c$) \& total-null ($\CMcal{N}_c \cap \CMcal{N}_e$) subspaces:} 
    Any power allocated to the Eve.-private subspace actively reduces the secrecy rate, as it contributes only to the eavesdropper's MI $R_e(\mathbf{F})$ without providing any benefit to the RX. Similarly, the power allocated to the total-null space is simply wasted, as it contributes to neither the RX's MI nor the eavesdropper's MI. Consequently, the optimal strategy allocates zero power to these directions.
\end{itemize}

This analysis naturally leads to a two-tiered power allocation strategy: allocate $O(P_{\mathrm{tot}})$ to the Comm.-private subspace for DoF gains and $o(P_{\mathrm{tot}})$ to the beneficial parts of the common subspace for constant gains. 
However, such an ideal power allocation strategy is infeasible,
because the four subspaces in \eqref{eq:me_spaces} are generally not mutually orthogonal. Consequently, there is no precoding basis that is both orthogonal and confined to a single subspace. 
For example, if power intended for a secure stream leaks into the Eve.-private subspace, such leakage directly penalizes the secrecy rate.

This challenge was addressed in \cite{KhistiMIMOME2010} by using the GSVD. Specifically, this method constructs a non-orthogonal basis that suitably controls the direction of the inherent power leakage.
In this design, the $o(P_{\mathrm{tot}})$ power allocated to the common subspace streams may leak into the Comm.-private subspace. However, since this subspace already carries a dominant $O(P_{\mathrm{tot}})$ power allocation, the leakage becomes asymptotically negligible in the high-SNR regime. 
This specific power hierarchy allows GSVD-based precoding to achieve the secrecy capacity in the high-SNR regime.

\subsection{MIMO-MS Channel} \label{sec:mimoms}
We now turn to the MIMO-MS channel, which corresponds to the special case where the eavesdropper is absent (i.e., $R_e(\mathbf{F}) = 0$). Analogous to the MIMO-ME analysis, we partition the transmit space $\mathbb{C}^{n_t}$ by considering the interplay between the row and null spaces of the communication channel ($\CMcal{R}_c, \CMcal{N}_c$) and the sensing channel ($\CMcal{R}_s, \CMcal{N}_s$). This yields a direct sum decomposition of four subspaces:
\begin{align}
    \mathbb{C}^{n_t} = 
    \underbrace{(\CMcal{V}_n^\perp\cap\CMcal{N}_s)}_{\text{Comm.-private}}
    \oplus
    \underbrace{(\CMcal{V}_n^\perp\cap\CMcal{N}_c)}_{\text{Sens.-private}}
    \oplus
    \underbrace{(\CMcal{R}_c\cap\CMcal{R}_s)}_{\text{common}}
    \oplus
    \underbrace{(\CMcal{N}_c\cap\CMcal{N}_s)}_{\text{total-null}},
    \label{eq:ms_spaces}
\end{align}
where $\CMcal{V}_n=\CMcal{N}_c\cap\CMcal{N}_s$ in this context. 
Similar to the MIMO-ME case, the decomposition \eqref{eq:ms_spaces} corresponds to a special case of the general framework established in Theorem~\ref{thm:mems_decomp} (Section~\ref{sec:mimomems}).
Unlike the MIMO-ME channel, where only the Comm.-private subspace ($\CMcal{V}_n^\perp\cap\CMcal{N}_e$) provides positive DoF gains in the high-SNR regime, 
in the MIMO-MS channel, the Comm.-private subspace ($\CMcal{V}_n^\perp\cap\CMcal{N}_s$), the Sens.-private subspace ($\CMcal{V}_n^\perp\cap\CMcal{N}_c$), and the common subspace ($\CMcal{R}_c\cap\CMcal{R}_s$) all contribute to the positive DoF. 
To achieve this, a dominant share of the power (i.e., $O(P_{\mathrm{tot}})$) should be allocated to each of the Comm.-private, Sens.-private, and common subspaces.

To further explore this, we first assume that the channels $\mathbf{H}_c$ and $\mathbf{H}_s$ share a common basis of right singular vectors, denoted by the unitary matrix $\mathbf{V}$. 
However, we clarify that this assumption does not hold in general. It is introduced to provide insights into the optimal precoder structure. The general case where $\mathbf{H}_c$ and $\mathbf{H}_s$ do not share the same right singular vectors will be discussed in Remark~\ref{remark:wmmse}.
Under this assumption, the Gram matrices $\mathbf{H}_c^{H}\mathbf{H}_c$ and $\mathbf{H}_s^{H}\mathbf{H}_s$ are simultaneously unitarily diagonalizable (i.e., they commute). Further, the four subspaces in \eqref{eq:ms_spaces} become mutually orthogonal, which eliminates inter-subspace power leakage and enables a decoupled analysis of the optimal precoder across individual subspaces.
To be specific, the shared eigenbasis $\mathbf{V}$ jointly diagonalizes
$\mathbf{H}_c^H\mathbf{H}_c$ and $\mathbf{H}_s^H\mathbf{H}_s$:
\begin{align}
    \mathbf{H}_c^H\mathbf{H}_c = \mathbf{V} \mathbf{\Lambda}_c \mathbf{V}^H, \quad  \mathbf{H}_s^H\mathbf{H}_s = \mathbf{V} \mathbf{\Lambda}_s \mathbf{V}^H, \label{eq:diagonalize}
\end{align}
where $\mathbf{\Lambda}_c = \diag(\lambda_{c,1}, \dots, \lambda_{c,n_t})$ and $\mathbf{\Lambda}_s = \diag(\lambda_{s,1}, \dots, \lambda_{s,n_t})$ are the diagonal matrices containing the respective channel eigenvalues. 

By Hadamard's inequality, selecting $\mathbf{V}$ as the precoding basis is optimal. Upon applying $\mathbf{V}$ as the precoding basis, the weighted sum maximization problem \eqref{eq:prob-log} reduces to: 
\begin{align}
    \max_{\{p_k\}} \quad &
    \sum_{k=1}^{N_s}\Bigl( w_c \log(1+\lambda_{c,k}p_k) + w_s \log(1+\lambda_{s,k}p_k) \Bigr) \nonumber \\
    \text{s.t.}\quad &
    p_k \ge 0,\ k=1,\ldots,N_s,\quad \sum_{k=1}^{N_s} p_k \le P_{\mathrm{tot}}.
\end{align}
where $p_k$ is the power allocated to the $k$-th eigenmode. 
This problem is a standard convex optimization problem; therefore, the solution derived from the Karush--Kuhn--Tucker (KKT) conditions is guaranteed to be the global optimum. For any eigenmode $k$ that receives non-zero power ($p_k > 0$), the stationarity condition requires that 
\begin{align}
    \frac{\partial \mathcal{L}}{\partial p_k} = \frac{w_c \lambda_{c,k}}{1 + \lambda_{c,k} p_k} + \frac{w_s \lambda_{s,k}}{1 + \lambda_{s,k} p_k} = \nu,
    \label{eq:kkt_condition_ms}
\end{align}
where $\nu$ is the Lagrange multiplier.
Solving \eqref{eq:kkt_condition_ms} for $p_k$ yields a generalized water-filling solution, where the water-level $\nu$ is chosen to satisfy the total power constraint. 
In the high-SNR regime ($P_{\mathrm{tot}}\to\infty$), power is allocated within the Comm.-private subspace, the Sens.-private subspace, and the common subspace, while the allocation across these subspaces is governed by the weights ($w_c$, $w_s$, $w_c+w_s$), respectively. 

In the general case where $\mathbf{H}_c$ and $\mathbf{H}_s$ have distinct bases, however, the above clean separation no longer holds. 
Since the subspaces in \eqref{eq:ms_spaces} are not orthogonal, the power allocated to one subspace may leak into other subspaces. 
For this reason, the high-SNR optimality condition, under which the Comm.-private, Sens.-private, and common subspaces must be allocated power proportional to $w_c$, $w_s$, and $w_c+w_s$, respectively, cannot be sustained.
This stands in sharp contrast to the MIMO-ME channel. 
In the MIMO-ME channel, the optimality condition can still be maintained despite power leakage, owing to the two-tiered power allocation ($O(P_{\mathrm{tot}})$ vs $o(P_{\mathrm{tot}})$). 
Since leakage from an $o(P_{\mathrm{tot}})$ stream into an $O(P_{\mathrm{tot}})$ stream is asymptotically negligible, such leakage does not hurt optimality. 
In the MIMO-MS channel, however, all active subspaces require $O(P_{\mathrm{tot}})$ power, rendering power leakage fundamentally detrimental.
Consequently, in the general case where the right singular vectors are not shared, a closed-form characterization of the optimal precoder structure is, unfortunately, no longer attainable.

\begin{remark}[On WMMSE-based optimization for MIMO-MS] \normalfont \label{remark:wmmse}
Although a closed-form solution for the MIMO-MS channel is generally intractable, WMMSE-based algorithms can be effectively employed due to the structural similarity between SMI and communication MI \cite{wangUnifiedISACPareto2024}. However, the standard WMMSE framework cannot be directly applied to the MIMO-ME-MS problem in~\eqref{eq:prob-log}. This is because the secrecy rate involves a difference of terms ($R_c - R_e$), which breaks the equivalence between the log-det rate and the weighted MSE required for convergence.
\end{remark}

Having analyzed the constituent MIMO-ME and MIMO-MS subproblems, 
we now address the MIMO-ME-MS channel, where the objectives of secure communication and sensing must be jointly optimized. 

\begin{table}[t]
    \centering
    \caption{Subspace Decomposition for MIMO-ME-MS} 
    \label{tab:subspace_dof}
    \begin{tabular}{@{}lll@{}}
    \toprule
    \textbf{Label} & \textbf{Definition} & \textbf{DoF Weight} \\
    \midrule
    $\CMcal{V}_{n}$   & $\CMcal{N}_c \cap \CMcal{N}_s \cap \CMcal{N}_e$ & $0$ \\ 
    $\CMcal{V}_{c}$   & $\CMcal{V}_n^\perp \cap\CMcal{N}_s \cap \CMcal{N}_e$ & $+w_c$ \\ 
    $\CMcal{V}_{s}$   & $\CMcal{V}_n^\perp \cap\CMcal{N}_c \cap \CMcal{N}_e$ & $+w_s$ \\
    $\CMcal{V}_{cs}$  & $(\bigoplus_{j \in \{n,c,s\}} \CMcal{V}_j)^\perp \cap \CMcal{N}_e$ & $w_c+w_s$ \\
    $\CMcal{V}_{cse}$ & $\CMcal{R}_c \cap \CMcal{R}_s \cap \CMcal{R}_e$ & $+w_s$ \\
    $\CMcal{V}_{se}$  & $\CMcal{V}_{cse}^\perp \cap\CMcal{R}_s \cap \CMcal{R}_e$ & $w_s-w_c$ \\
    $\CMcal{V}_{ce}$  & $\CMcal{V}_{cse}^\perp \cap\CMcal{R}_c \cap \CMcal{R}_e$ & $0$ \\
    $\CMcal{V}_{e}$   & $(\bigoplus_{j \in \{se,ce,cse\}} \CMcal{V}_j)^\perp \cap \CMcal{R}_e$ & $-w_c$ \\
    \bottomrule
    \end{tabular}
\end{table}

\subsection{MIMO-ME-MS Channel} \label{sec:mimomems}
Finally, we partition the transmit space $\mathbb{C}^{n_t}$ by considering the interplay among the row and null spaces of all three effective channels: $\mathbf{H}_c$, $\mathbf{H}_s$, and $\mathbf{H}_e$. This yields a complete decomposition of the transmit space into a direct sum of eight subspaces, as summarized in Theorem~\ref{thm:mems_decomp} and Table~\ref{tab:subspace_dof}.
We note that Theorem~\ref{thm:mems_decomp} encompasses the subspace partitioning results for the MIMO-ME and MIMO-MS channels presented in \eqref{eq:me_spaces} and \eqref{eq:ms_spaces}.

\begin{theorem}[Subspace decomposition] \label{thm:mems_decomp}
The eight subspaces $\{\CMcal{V}_j\}$ defined in Table~\ref{tab:subspace_dof} form a direct sum decomposition of the transmit space $\mathbb{C}^{n_t}$:
\begin{align}
    \mathbb{C}^{n_t} = \bigoplus_{j \in \mathcal{K}} \CMcal{V}_{j}, \label{eq:mems_spaces}
\end{align}
where $\mathcal{K} = \{n, c, s, e, cs, se, ce, cse\}$. 
\end{theorem}

\begin{IEEEproof}
    The proof is provided in Appendix~\ref{sec:appen1}.
\end{IEEEproof}

The precise definition of each subspace and its corresponding DoF contribution are summarized in Table~\ref{tab:subspace_dof}. 
As observed in the previous analysis of the MIMO-ME and MIMO-MS channels, the optimal precoder should be structured to incorporate the DoF gains offered by each subspace: 
\begin{itemize}
    \item \textbf{Positive DoF gain ($\CMcal{V}_{c}$, $\CMcal{V}_{s}, \CMcal{V}_{cs}$, $\CMcal{V}_{cse}$):}
    The private subspaces \textbf{$\CMcal{V}_{c}$} (Comm.-private) and \textbf{$\CMcal{V}_{s}$} (Sens.-private) provide positive DoF gains of $w_c$ and $w_s$, respectively. The full-common space \textbf{$\CMcal{V}_{cse}$} is dominated by the sensing objective, providing a DoF of $w_s$. The Secure-ISAC space \textbf{$\CMcal{V}_{cs}$} is the most beneficial, offering a combined DoF gain of $w_c+w_s$. To achieve the logarithmic rate gains, a dominant $O(P_{\mathrm{tot}})$ power allocation across the corresponding subspaces is necessary. 

    \item \textbf{Conditional or synergistic gain ($\CMcal{V}_{se}$, $\CMcal{V}_{ce}$):} 
    These subspaces introduce intricate tradeoffs but can be jointly exploited via subspace alignment. Pairing up to $\min\{k_{ce}, k_{se}\}$ dimensions from both yields a positive net DoF of $w_s$, requiring $O(P_{\mathrm{tot}})$ power. The remaining unaligned dimensions retain their individual characteristics. 
    The leftover Sens.-eve.-common subspace \textbf{$\CMcal{V}_{se}$} yields a net DoF of $w_s - w_c$. Its optimal power allocation depends on the weights: $O(P_{\mathrm{tot}})$ if $w_s > w_c$, zero if $w_s < w_c$ (to avoid rate penalties), and $o(P_{\mathrm{tot}})$ if $w_s = w_c$ to capture constant gain.
    The leftover Comm.-eve.-common space \textbf{$\CMcal{V}_{ce}$} offers zero DoF, requiring only vanishingly small power, $o(P_{\mathrm{tot}})$, to exploit its constant gain.
    
    \item \textbf{Negative or zero gain ($\CMcal{V}_{e}$, $\CMcal{V}_{n}$):}
    Any power in the Eve.-private space $\CMcal{V}_{e}$ incurs a rate penalty with a DoF of $-w_c$, while power in the total-null space $\CMcal{V}_{n}$ is wasted. Consequently, the optimal strategy allocates zero power to both.
\end{itemize}

This subspace analysis reveals why the MIMO-ME-MS channel is fundamentally more challenging than a simple superposition of its constituent parts. A natural question is whether optimal strategies for the subproblems can be combined. For instance, consider an ideal setting where the communication and sensing channels are aligned (i.e., share a common eigenbasis).
In this special case, the optimal precoder for the MIMO-MS channel is known to be constructed from the shared eigenbasis $\mathbf{V}$ in \eqref{eq:diagonalize}.
One might therefore conjecture that a straightforward combination of design principles that are optimal for the MIMO-MS and MIMO-ME channels, namely, the GSVD-based precoding with the common eigenbasis, would also be asymptotically optimal for the MIMO-ME-MS channel. However, such a combination is strictly suboptimal.
The effectiveness of GSVD-based precoding in the MIMO-ME scenario critically depends on a two-tiered power allocation strategy: dominant $O(P_{\mathrm{tot}})$ power for the private subspace and vanishing $o(P_{\mathrm{tot}})$ power for the common subspace. This ensures that any power leakage from the low-power common-space streams into the high-power private-space streams becomes asymptotically negligible ($O(P_{\mathrm{tot}})$ vs. $o(P_{\mathrm{tot}})$) and is therefore harmless to secrecy performance. 
This premise is violated in the MIMO-ME-MS channel. Here, achieving the maximum DoF requires allocating $O(P_{\mathrm{tot}})$ power to multiple subspaces, including the common space $\CMcal{V}_{cse}$. 
That is to say, in the MIMO-ME-MS channel, streams in both the common and private subspaces must carry $O(P_{\mathrm{tot}})$ power to maximize DoF. 
Consequently, when power allocated to a common-space stream leaks into a private-space stream, this high-power leakage is on the same order as the intended power of the private-space stream and is thus no longer asymptotically negligible. 
This leakage breaks the fragile decoupling of streams that underpins the optimality of GSVD-based precoding, rendering the independent power allocation across streams intractable. 
Therefore, the MIMO-ME-MS channel invalidates the direct extension of existing optimal schemes, even under idealized conditions (i.e., shared eigenbasis), and necessitates a more robust design framework capable of managing these new high-power leakage pathways.

Having established that the optimal precoder structure is intractable, 
we turn our attention to a quasi-optimal precoder structure. 
A precoder is defined as {\emph{quasi-optimal}} if it achieves the maximum possible weighted DoF. This implies that the quasi-optimal precoder achieves the optimum performance with a constant gap that does not scale with SNR. 
To this end, we first derive an upper bound on the weighted DoF. 

To determine this upper bound, we leverage the unique decomposition of any precoder $\mathbf{F}$ based on the direct sum structure in \eqref{eq:mems_spaces}. 
Crucially, while the subspaces $\CMcal{V}_j$ are not, in general, mutually orthogonal, the fact that they form a direct sum is sufficient to establish a tight upper bound on the achievable DoF. This result is formalized in the following theorem.

\begin{theorem}[Upper bound of weighted DoF] \label{thm:dof_max}
The weighted DoF $d(\mathbf{F})$ for any precoder $\mathbf{F}$ is upper-bounded by $d_{\text{max}}$, defined as:
\begin{align}
    d_{\text{max}} &\triangleq w_c k_c + w_s k_s + (w_c+w_s)k_{cs} + w_s k_{cse} \nonumber\\
    &\quad + w_s k_{se} - \min\{w_c, w_s\}[k_{se}-k_{ce}]^+,
\end{align}
where $k_j \triangleq \dim(\CMcal{V}_j)$ for $j \in \mathcal{K}$.
\end{theorem}

\begin{IEEEproof}
The proof is provided in Appendix~\ref{sec:appen2}.
\end{IEEEproof}

To achieve this bound, a precoder must be designed to exclusively activate the subspaces that contribute positively to the weighted DoF. This includes the aligned subspace $\CMcal{W}_{\text{align}}$ formed by pairing up to $k_{w} = \min\{k_{ce}, k_{se}\}$ basis vectors of $\CMcal{V}_{ce}$ and $\CMcal{V}_{se}$. Let us define the \emph{useful subspace}, $\CMcal{V}_{\text{useful}}$, as the direct sum of these contributing subspaces:
\begin{align}
    \hspace{-0.2em} \CMcal{V}_{\text{useful}} \hspace{-0.1em}\triangleq\hspace{-0.1em} \CMcal{V}_c \hspace{-0.1em}\oplus\hspace{-0.1em} \CMcal{V}_s \hspace{-0.1em}\oplus\hspace{-0.1em} \CMcal{V}_{cs} \hspace{-0.1em}\oplus\hspace{-0.1em} \CMcal{V}_{cse} \hspace{-0.1em}\oplus\hspace{-0.1em} \CMcal{W}_{\text{align}} \hspace{-0.1em}\oplus\hspace{-0.1em} \begin{cases} \CMcal{V}_{se}^{\text{rem}} & \hspace{-0.3em}\text{if } w_s > w_c \\ \{\mathbf{0}\} & \hspace{-0.3em}\text{otherwise} \end{cases},
\end{align}
where $\CMcal{V}_{se}^{\text{rem}} \subseteq \CMcal{V}_{se}$ denotes the unaligned remaining portion of dimension $k_{rem}=[k_{se}-k_{ce}]^+$.

Based on this, the sufficient conditions under which a precoder $\mathbf{F}$ is quasi-optimal are presented as follows: 
\begin{enumerate}
    \item Its column space lies entirely within this useful subspace, i.e., $\CMcal{C}(\mathbf{F}) \subseteq \CMcal{V}_{\text{useful}}$.
    \item Its column space spans the entire useful subspace, i.e., $\rank(\mathbf{F}) = \dim(\CMcal{V}_{\text{useful}})$, and its power allocation is non-degenerate, meaning all of its singular values scale as $\sqrt{P_{\mathrm{tot}}}$.
\end{enumerate}
These conditions imply that to attain the maximum weighted DoF, a precoder must not only span the useful subspace $\CMcal{V}_{\text{useful}}$ but also distribute its power non-degenerately across all of its dimensions, while simultaneously avoiding the harmful subspaces like $\CMcal{V}_e$, $\CMcal{V}_n$, and the unused portion of $\CMcal{V}_{ce}$.

\begin{figure}[t]
  \centering
  \includegraphics[width=0.75\columnwidth]{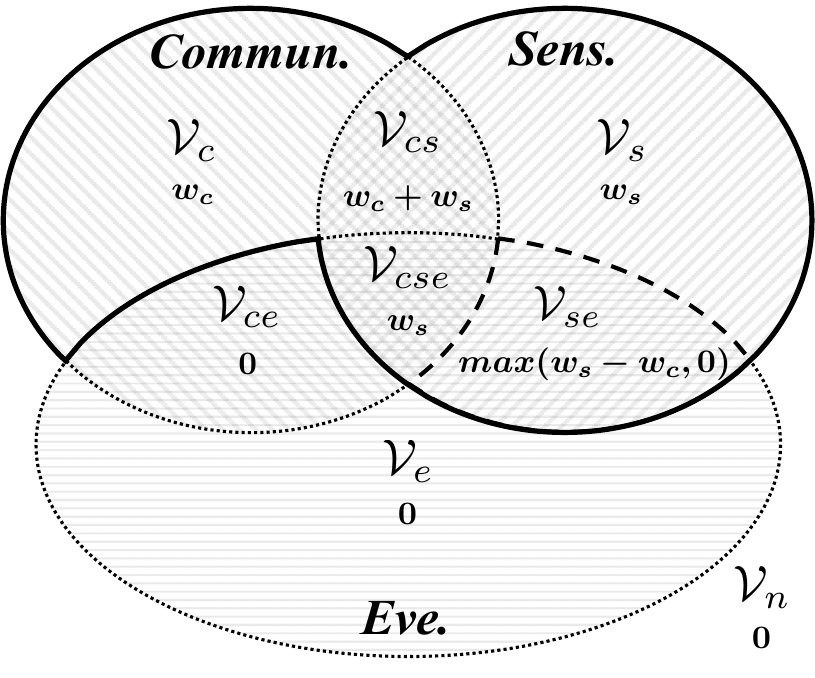}
  \caption{A Venn diagram illustrating the effective DoF weight for each subspace. A quasi-optimal precoder must allocate dominant power to span the regions with positive weights, which constitute the useful space $\CMcal{V}_{\text{useful}}$.}
  \label{fig:DoF_landscape}
\end{figure}

To facilitate understanding of the quasi-optimal precoder structure, Fig.~\ref{fig:DoF_landscape} illustrates the asymptotic DoF gains associated with each subspace. 
The regions with positive weights, along with the aligned subspace $\CMcal{W}_{\text{align}}$, collectively constitute the useful subspace $\CMcal{V}_{\text{useful}}$, which a quasi-optimal precoder must span by allocating dominant power on the order of $O(P_{\mathrm{tot}})$. 
In contrast, the subspaces corresponding to regions with zero or negative weights must be nullified by the precoder to avoid power leakage into directions that do not contribute to or actively penalize the DoF.
Subsequently, we prove that the proposed precoder structure is able to achieve the DoF upper bound.
\begin{proposition}[Achievability of the DoF upper bound] \normalfont
A precoder $\mathbf{F}_{\text{q-opt}}$ satisfying conditions 1) and 2) achieves the DoF upper bound $d_{\text{max}}$.
\end{proposition}

\begin{IEEEproof}
    Since $\CMcal{C}(\mathbf{F}_{\text{q-opt}}) = \CMcal{V}_{\text{useful}}$, and its power is allocated non-degenerately, the rank terms are calculated by summing the dimensions of the constituent subspaces of $\CMcal{V}_{\text{useful}}$ that lie within each channel's row space. Incorporating the subspace alignment between $\CMcal{V}_{ce}$ and $\CMcal{V}_{se}$, which contributes $\min\{k_{ce}, k_{se}\}$ dimensions to all receivers, alongside the leftover $\CMcal{V}_{se}$ dimensions when $w_s > w_c$, we obtain:
\begin{align}
    \rank(\mathbf{H}_c\mathbf{F}_{\text{q-opt}}) \hspace{-0.1em}&=\hspace{-0.1em} k_c + k_{cs} + k_{cse} + k_{w}, \\
    \rank(\mathbf{H}_s\mathbf{F}_{\text{q-opt}}) \hspace{-0.1em}&=\hspace{-0.1em} k_s + k_{cs} + k_{cse} + k_{w} + \mathbb{I}(w_s \hspace{-0.1em}>\hspace{-0.1em} w_c) k_{rem}, \\
    \rank(\mathbf{H}_e\mathbf{F}_{\text{q-opt}}) \hspace{-0.1em}&=\hspace{-0.1em} k_{cse} + k_{w} + \mathbb{I}(w_s \hspace{-0.1em}>\hspace{-0.1em} w_c) k_{rem},
\end{align}
where $\mathbb{I}(\cdot)$ denotes the indicator function.
Substituting these evaluations into the weighted DoF formula $d(\mathbf{F}) = w_c \rank(\mathbf{H}_c\mathbf{F}) - w_c \rank(\mathbf{H}_e\mathbf{F}) + w_s \rank(\mathbf{H}_s\mathbf{F})$ and simplifying the terms directly yields $d(\mathbf{F}_{\text{q-opt}}) = d_{\text{max}}$. 
\end{IEEEproof}

It is worth noting that the proposed quasi-optimal precoder achieves the optimal weighted DoF in the high-SNR regime. However, when considering the low-SNR regime, the design criteria may differ substantially, as maximizing DoF is no longer aligned with performance optimization under power-limited conditions. 
Namely, the rate function exhibits approximately linear behavior (i.e., $\log(1+x) \approx x/\ln 2$ for small $x$). 
Applying this in \eqref{eq:prob-log}, the problem simplifies to: 
\begin{align}
    \max_{\mathbf{F}} \quad & \tr\left(\mathbf{F}^H \left(w_c \mathbf{H}_c^H \mathbf{H}_c - w_c \mathbf{H}_e^H \mathbf{H}_e + w_s \mathbf{H}_s^H \mathbf{H}_s\right) \mathbf{F}\right) \\
    \text{s.t.} \quad & \tr(\mathbf{F}\,\mathbf{F}^H) \le P_{\mathrm{tot}}. \nonumber
\end{align}
The well-known solution to this problem is a rank-one precoder, where all power is allocated to the direction of the principal eigenvector of the composite matrix $\mathbf{M} = w_c \mathbf{H}_c^H \mathbf{H}_c - w_c \mathbf{H}_e^H \mathbf{H}_e + w_s \mathbf{H}_s^H \mathbf{H}_s$.

These analyses of the high- and low-SNR regimes reveal two distinct design principles. A quasi-optimal high-SNR precoder must be spatially expansive, allocating power across all dimensions of the useful subspace $\CMcal{V}_{\text{useful}}$ to achieve the maximum DoF. 
Conversely, an optimal low-SNR precoder must be spatially focused, concentrating the entire power budget into the single most effective beamforming direction. 
Accordingly, a practical precoder that operates reliably across a wide range of SNR regimes must possess the flexibility to seamlessly interpolate between these two contrasting behaviors: spatial expansiveness at high SNR and spatial focus at low SNR. 
To address this, the following section presents a practical precoding approach that incorporates the structural insights drawn from both asymptotic regimes.

\section{Practical Precoder Design}\label{sec:algorithm}

In this section, we propose a two-stage iterative algorithm that alternates between basis construction and power allocation to solve \eqref{eq:prob-log}. 
In the first stage, given a fixed power allocation, we sequentially construct a new precoding basis $\mathbf{W}$. In the second stage, given the updated basis, we optimize the power allocation $\mathbf{P}$. The algorithm alternates between these two stages until the objective function in \eqref{eq:prob-log} converges.
We explain the detailed process as follows.

\subsection{Rate Decomposition and Reformulation}

We first introduce a sequential rate-decomposition technique for constructing the basis vectors. A related decomposition was also used in \cite{GaoHybrid2016}.
\begin{proposition}\label{prop:greedy}
    For $i \in \{c, s, e\}$, the rate $R_i(\mathbf{F})$ can be decomposed into a sum of marginal gains from each sequentially constructed stream. For $\mathbf{W}=[\mathbf{w}_1, \dots, \mathbf{w}_{N_s}]$ and $\mathbf{P}=\diag(p_1, \dots, p_{N_s})$ with $\mathbf{F} = \mathbf{W} \mathbf{P}^{1/2}$, we have:
    \begin{align}
        R_i(\mathbf{F}) = \sum_{n=1}^{N_s} \log \left(1+ p_n \mathbf{w}_{n}^{H} \mathbf{G}^{(i)}_{n-1} \mathbf{w}_{n}\right), \label{eq:rate_decompose}
    \end{align}
    where $\mathbf{G}^{(i)}_{n-1}=\mathbf{H}_i^H \left(\mathbf{T}^{(i)}_{n-1}\right)^{-1} \mathbf{H}_i$ is the effective gain matrix for stream $n$ after accounting for the first $n-1$ streams. 
    Here, $\mathbf{T}^{(i)}_{n-1} = \mathbf{I} + \mathbf{H}_i \mathbf{W}_{n-1} \mathbf{P}_{n-1} \mathbf{W}_{n-1}^{H} \mathbf{H}_i^{H}$ represents the residual covariance matrix from the first $n-1$ streams, with $\mathbf{W}_{n-1} = [\mathbf{w}_1, \dots, \mathbf{w}_{n-1}]$ and $\mathbf{P}_{n-1} = \diag(p_1, \dots, p_{n-1})$.
\end{proposition}

\begin{IEEEproof}
Split $\mathbf{W}\mathbf{P}\mathbf{W}^H$ into the first $N_s-1$ streams and the $N_s$-th stream as $\mathbf{W}\mathbf{P}\mathbf{W}^H = \mathbf{W}_{N_s-1}\mathbf{P}_{N_s-1}\mathbf{W}_{N_s-1}^H + p_{N_s} \mathbf{w}_{N_s} \mathbf{w}_{N_s}^H$. Then the total rate for link $i$ can be written as:
\begin{align}
    & R_{i}({\mathbf{F}}) = \log\det(\mathbf{T}^{(i)}_{N_s}) = \log\det(\mathbf{I} + \mathbf{H}_i \mathbf{W}\mathbf{P}\mathbf{W}^H \mathbf{H}_i^H) \nonumber \\
    &\hspace{-0.2em}= \log\det(\underbrace{\mathbf{I} + \mathbf{H}_i \mathbf{W}_{N_s-1}\mathbf{P}_{N_s-1}\mathbf{W}_{N_s-1}^H \mathbf{H}_i^H}_{\mathbf{T}^{(i)}_{N_s-1}} 
    \hspace{-0.3em} + p_{N_s} \mathbf{H}_i \mathbf{w}_{N_s} \mathbf{w}_{N_s}^H \mathbf{H}_i^H).
\end{align}
Applying the matrix determinant lemma, we get:
\begin{align}
    R_{i}(\mathbf{F}) = \log\det(\mathbf{T}^{(i)}_{N_s-1}) + \log(1 + p_{N_s} \mathbf{w}_{N_s}^H \mathbf{G}^{(i)}_{N_s-1} \mathbf{w}_{N_s} ).
\end{align}
The term $\log\det(\mathbf{T}^{(i)}_{N_s-1})$ represents the rate from the first $N_s-1$ streams. 
Recursively applying this decomposition from $n=N_s$ down to $1$ yields \eqref{eq:rate_decompose}, with $\mathbf{G}^{(i)}_{n-1} = \mathbf{H}_i^H (\mathbf{T}^{(i)}_{n-1})^{-1} \mathbf{H}_i$.
\end{IEEEproof}

Proposition~\ref{prop:greedy} characterizes the marginal contribution of each sequentially added stream.
Recalling our objective function in \eqref{eq:prob-log}, we apply this decomposition to each rate term. Consequently, the net marginal gain from adding the $n$-th basis vector $\mathbf{w}_n$ with power $p_n$ is given by:
\begin{align}
    R_n \hspace{-0.1em} &= w_c \log \left(1\hspace{-0.05em}+\hspace{-0.05em} p_n \mathbf{w}_{n}^{H} \mathbf{G}^{(c)}_{n-1} \mathbf{w}_{n}\right) \hspace{-0.07em}-\hspace{-0.07em} w_c \log \left(1\hspace{-0.05em}+\hspace{-0.05em} p_n \mathbf{w}_{n}^{H} \mathbf{G}^{(e)}_{n-1} \mathbf{w}_{n}\right) \nonumber \\ 
    &\quad +w_s \log \left(1\hspace{-0.05em}+\hspace{-0.05em} p_n \mathbf{w}_{n}^{H} \mathbf{G}^{(s)}_{n-1} \mathbf{w}_{n}\right). \label{eq:net_marginal}
\end{align}
This decomposition allows us to formulate a tractable subproblem for finding the next basis vector at each step.
Given previously obtained basis vectors $\{\mathbf{w}_1, \dots, \mathbf{w}_{n-1}\}$, we seek a vector ${\mathbf{f}}=\mathbf{w}_{n}$ that maximizes \eqref{eq:net_marginal}.
The vector ${\mathbf{f}}$ must be unit-norm (${\mathbf{f}}^{H} {\mathbf{f}} =1$) and orthogonal to all previously obtained basis vectors ($\mathbf{W}_{n-1}^H\mathbf{f}=\mathbf{0}$). 
Under $\mathbf{f}^H\mathbf{f}=1$, each term satisfies:
\begin{align}
    \log(1+ p_n \mathbf{f}^{H} \mathbf{G}^{(i)}_{n-1} \mathbf{f}) = \log(\mathbf{f}^H(\mathbf{I} + p_n\mathbf{G}^{(i)}_{n-1})\mathbf{f}).
\end{align}
Define $\mathbf{A}^{(i)}_{n-1} \triangleq \mathbf{I} + p_n \mathbf{G}^{(i)}_{n-1}$. Then the basis-update subproblem becomes:
\begin{align}
    \max_{\mathbf{f}} \quad & w_c \log\left(\frac{ \mathbf{f}^H \mathbf{A}^{(c)}_{n-1} \mathbf{f}}{\mathbf{f}^H \mathbf{A}^{(e)}_{n-1} \mathbf{f}}\right) + w_s \log\bigl(\mathbf{f}^H \mathbf{A}^{(s)}_{n-1} \mathbf{f}\bigr) \label{eq:prob_constrained} \\
    \text{s.t.} \quad & \mathbf{f}^H \mathbf{f} = 1, \quad \mathbf{W}_{n-1}^H\mathbf{f}=\mathbf{0}. \nonumber
\end{align}
This problem remains challenging due to the nonconvex objective and the two constraints. 

We next reformulate the problem to eliminate these constraints. Given the constraint $\mathbf{f}^H\mathbf{f}=1$, the sensing term can be rewritten as $\log(\mathbf{f}^H \mathbf{A}^{(s)}_{n-1} \mathbf{f} / \mathbf{f}^H \mathbf{I} \mathbf{f})$ for any feasible solution. Consequently, the objective function becomes scale-invariant with respect to $\mathbf{f}$. Accordingly, the unit-norm constraint can be omitted.
Next, the orthogonality constraint $\mathbf{W}_{n-1}^{H}\mathbf{f} = \mathbf{0}$ implies $\mathbf{f} \in \CMcal{N}(\mathbf{W}_{n-1}^{H})$. Let $\mathbf{\Pi}_{n-1}$ denote the orthogonal projection matrix onto $\CMcal{N}(\mathbf{W}_{n-1}^{H})$. Then, the constraint is equivalently enforced by $\mathbf{f} = \mathbf{\Pi}_{n-1} \mathbf{f}$.
Substituting this into any quadratic form $\mathbf{f}^{H}\mathbf{M}\mathbf{f}$ yields $\mathbf{f}^{H}(\mathbf{\Pi}_{n-1}\mathbf{M}\mathbf{\Pi}_{n-1})\mathbf{f}$. 
Leveraging this property, we incorporate the orthogonality constraint by replacing $\mathbf{A}_{n-1}^{(i)}$ and $\mathbf{I}$ with their projected versions, $\tilde{\mathbf{A}}_{n-1}^{(i)} \triangleq \mathbf{\Pi}_{n-1}\mathbf{A}_{n-1}^{(i)}\mathbf{\Pi}_{n-1}$ and $\tilde{\mathbf{I}}_{n-1} \triangleq \mathbf{\Pi}_{n-1}\mathbf{I}\mathbf{\Pi}_{n-1} = \mathbf{\Pi}_{n-1}$, respectively. 
This allows us to omit the explicit orthogonality constraint without loss of optimality.

These steps yield the following unconstrained formulation:
\begin{align}
    \max_{\mathbf{f}} \quad & w_c \log\left(\frac{ \mathbf{f}^H \tilde{\mathbf{A}}^{(c)}_{n-1} \mathbf{f}}{\mathbf{f}^H \tilde{\mathbf{A}}^{(e)}_{n-1} \mathbf{f}}\right) + w_s \log\left(\frac{\mathbf{f}^H \tilde{\mathbf{A}}^{(s)}_{n-1} \mathbf{f}}{\mathbf{f}^H \tilde{\mathbf{I}}_{n-1} \mathbf{f}}\right). \label{eq:prob-basis}
\end{align} 
The problem in \eqref{eq:prob-basis} forms the core of our sequential basis construction procedure. The process for solving it is detailed in the next subsection.

\subsection{Stage 1: Basis Vector Update} 

The first stage aims to find the updated precoding basis $\mathbf{W}$ under a fixed power allocation $\mathbf{P}$. Based on the sequential rate decomposition in Proposition~\ref{prop:greedy}, the basis vectors are computed one by one. For each $n \in \{1, \dots, N_s\}$, the optimal direction $\mathbf{w}_n$ is found by solving \eqref{eq:prob-basis}. 
As problem \eqref{eq:prob-basis} is nonconvex, we find a stationary point by analyzing its first-order optimality conditions. For notational simplicity in this derivation, we drop the subscript $n-1$ from the matrices. The objective is:
\begin{align}
    J(\mathbf{f}) = w_c \log\left(\frac{ \mathbf{f}^H \tilde{\mathbf{A}}^{(c)} \mathbf{f}}{\mathbf{f}^H \tilde{\mathbf{A}}^{(e)} \mathbf{f}}\right) + w_s \log\left(\frac{\mathbf{f}^H \tilde{\mathbf{A}}^{(s)} \mathbf{f}}{\mathbf{f}^H \tilde{\mathbf{I}} \mathbf{f}}\right).
\end{align}
The first-order stationarity condition is given by $\nabla_{\mathbf{f}^*}J(\mathbf{f})=\mathbf{0}$. By evaluating this gradient, we obtain:
\begin{align}
    \underbrace{ \left( \frac{w_c \tilde{\mathbf{A}}^{(c)}}{\mathbf{f}^H \tilde{\mathbf{A}}^{(c)} \mathbf{f}} + \frac{w_s \tilde{\mathbf{A}}^{(s)}}{\mathbf{f}^H \tilde{\mathbf{A}}^{(s)} \mathbf{f}} \right) }_{\mathbf{B}(\mathbf{f})} \mathbf{f} = \underbrace{ \left( \frac{w_c \tilde{\mathbf{A}}^{(e)}}{\mathbf{f}^H \tilde{\mathbf{A}}^{(e)} \mathbf{f}} + \frac{w_s \tilde{\mathbf{I}}}{\mathbf{f}^H \tilde{\mathbf{I}} \mathbf{f}} \right) }_{\mathbf{C}(\mathbf{f})} \mathbf{f}. \label{eq:KKT_cond}
\end{align}
We note that this equation takes the form $\mathbf{B}(\mathbf{f})\mathbf{f} = \mathbf{C}(\mathbf{f})\mathbf{f}$, where both matrices depend on the vector $\mathbf{f}$. 
Due to this coupling, a closed-form solution is generally intractable. 
To resolve this, we adopt a fixed-point iteration. This approach has also been explored as generalized power iteration in the context of MIMO rate maximization \cite{parkGPIRS23}. Building on this method, we construct our update process as follows: 
\begin{align}
    \mathbf{f} \leftarrow \frac{\left( \mathbf{C}(\mathbf{f}) \right)^{\dagger} \mathbf{B}(\mathbf{f}) \mathbf{f}}{\|\left( \mathbf{C}(\mathbf{f}) \right)^{\dagger} \mathbf{B}(\mathbf{f}) \mathbf{f}\|_2} . \label{eq:iter_matrix}
\end{align}
We repeat \eqref{eq:iter_matrix} until convergence. 

To reduce the computational load and ensure stable convergence across outer iterations, we employ a warm-start strategy for the fixed-point iteration in \eqref{eq:iter_matrix}. 
Specifically, we set the initial iterate $\mathbf{f}^{(0)}$ for stream $n$ using the $n$-th column of $\mathbf{W}^{\mathrm{prev}}$, i.e., the basis matrix obtained at the previous outer iteration.
This vector is projected using the current null-space projector $\mathbf{\Pi}$ to ensure feasibility.
This approach ensures that the basis construction starts in the vicinity of the stationary point found in the previous outer iteration.
Upon convergence, the algorithm yields a vector $\mathbf{w}_n$ that satisfies the first-order optimality conditions, thus providing a stationary solution for the basis construction subproblem. 

Once $\mathbf{w}_n$ is obtained, we update the effective gain matrices using Proposition~\ref{prop:greedy_update}.
To reduce the computational complexity of this process, we present the following proposition, which avoids full matrix inversion. 

\begin{proposition}\label{prop:greedy_update}
    Define the projected effective gain matrix
    \begin{align}
        \tilde{\mathbf{G}}^{(i)}_n \triangleq \mathbf{\Pi}_n \mathbf{G}^{(i)}_n \mathbf{\Pi}_n,\qquad i\in\{c,s,e\}.
    \end{align}
    Then $\tilde{\mathbf{G}}^{(i)}_n$ can be updated recursively from $\tilde{\mathbf{G}}^{(i)}_{n-1}$ without a full matrix inversion. Given the $n$-th basis vector $\mathbf{w}_n$ and its power $p_n$, the update is:
    \begin{align}
        \tilde{\mathbf{G}}^{(i)}_{n} = \mathbf{\Pi}_{n}
        \Bigl(\tilde{\mathbf{G}}^{(i)}_{n-1} - \frac{(\tilde{\mathbf{G}}^{(i)}_{n-1}\mathbf{w}_{n})(\tilde{\mathbf{G}}^{(i)}_{n-1}\mathbf{w}_{n})^{H}} {1/p_{n} + \mathbf{w}_{n}^{H}\tilde{\mathbf{G}}^{(i)}_{n-1}\mathbf{w}_{n}} \Bigr) 
        \mathbf{\Pi}_{n}. \label{eq:G_update}
    \end{align}
\end{proposition}

\begin{IEEEproof}
The matrix $\mathbf{T}^{(i)}_{n}$ at step $n$ is a rank-1 update of the matrix at step $n-1$:
\begin{align}
    \mathbf{T}^{(i)}_{n} = \mathbf{T}^{(i)}_{n-1} + p_n \mathbf{H}_i \mathbf{w}_n \mathbf{w}_n^H \mathbf{H}_i^H.
\end{align}
Applying the Sherman-Morrison formula and then pre- and post-multiplying by $\mathbf{H}_i^H$ and $\mathbf{H}_i$, respectively, yields:
\begin{align}
    \mathbf{G}^{(i)}_{n} = \mathbf{G}^{(i)}_{n-1} - \frac{ \mathbf{G}^{(i)}_{n-1} \mathbf{w}_n \mathbf{w}_n^H \mathbf{G}^{(i)}_{n-1} }{1/p_n + \mathbf{w}_n^H \mathbf{G}^{(i)}_{n-1} \mathbf{w}_n}. \label{eq:G_update_simple}
\end{align}
Next, project \eqref{eq:G_update_simple} onto $\CMcal{N}(\mathbf{W}_n^H)$. Since $\mathbf{w}_n \in \CMcal{N}(\mathbf{W}_{n-1}^H)$, we have $\mathbf{\Pi}_{n-1}\mathbf{w}_n = \mathbf{w}_n$. Moreover, $\CMcal{N}(\mathbf{W}_n^H) \subseteq \CMcal{N}(\mathbf{W}_{n-1}^H)$ implies $\mathbf{\Pi}_{n}=\mathbf{\Pi}_{n}\mathbf{\Pi}_{n-1}$. Then we get:
\begin{align}
    \tilde{\mathbf{G}}^{(i)}_{n} &= \mathbf{\Pi}_{n} \mathbf{\Pi}_{n-1}\left( \mathbf{G}^{(i)}_{n-1} - \frac{ \mathbf{G}^{(i)}_{n-1} \mathbf{w}_n \mathbf{w}_n^H \mathbf{G}^{(i)}_{n-1} }{1/p_n + \mathbf{w}_n^H \mathbf{G}^{(i)}_{n-1} \mathbf{w}_n} \right) \mathbf{\Pi}_{n-1} \mathbf{\Pi}_{n} \nonumber \\
    &= \mathbf{\Pi}_{n} \left( \tilde{\mathbf{G}}^{(i)}_{n-1} - \frac{ \tilde{\mathbf{G}}^{(i)}_{n-1} \mathbf{w}_n \mathbf{w}_n^H \tilde{\mathbf{G}}^{(i)}_{n-1} }{1/p_n + \mathbf{w}_n^H \tilde{\mathbf{G}}^{(i)}_{n-1} \mathbf{w}_n} \right) \mathbf{\Pi}_{n}.
\end{align}
\end{IEEEproof}

By using Proposition~\ref{prop:greedy_update}, we can compute $\tilde{\mathbf{A}}^{(i)}_{n}$ efficiently as $\tilde{\mathbf{A}}^{(i)}_{n} = \tilde{\mathbf{I}}_{n} + p_{n+1} \tilde{\mathbf{G}}^{(i)}_{n}$. The resulting procedure is detailed in Algorithm \ref{alg:basisUpdate}.

\begin{algorithm}[t]
    \caption{Iterative Basis Vector Update}
    \label{alg:basisUpdate}
    \KwInput{Channel matrices $\mathbf{H}_c, \mathbf{H}_s, \mathbf{H}_e$; power allocation vector $\mathbf{p}$; weights $w_c, w_s$; number of streams $N_s$; tolerance $\varepsilon$; previous basis matrix $\mathbf{W}^{\mathrm{prev}}=[\mathbf{w}^{\mathrm{prev}}_1, \dots, \mathbf{w}^{\mathrm{prev}}_{N_s}]$}
    \KwOutput{Basis matrix $\mathbf{W} = [\mathbf{w}_1, \dots, \mathbf{w}_{N_s}]$}
    
    Initialize $\mathbf{\Pi} \leftarrow \mathbf{I}$, $\tilde{\mathbf{G}}^{(i)} \leftarrow \mathbf{H}_i^H \mathbf{H}_i$ for $i\in \{c,s,e\}$.
    
    \For{$n \leftarrow 1$ \KwTo $N_s$}{
        Initialize $\mathbf{f}^{(0)} \leftarrow \mathbf{\Pi}\mathbf{w}^{\mathrm{prev}}_n / \|\mathbf{\Pi} \mathbf{w}^{\mathrm{prev}}_n\|_2$ and $k \leftarrow 0$. \\
        Update $\tilde{\mathbf{I}} \leftarrow \mathbf{\Pi}$, $\tilde{\mathbf{A}}^{(i)} \leftarrow \tilde{\mathbf{I}} + p_n \tilde{\mathbf{G}}^{(i)}$ for $i\in \{c,s,e\}$.
        
        \Repeat{$\|\mathbf{f}^{(k)} - \mathbf{f}^{(k-1)}\|_2 < \varepsilon$}{
            $k \leftarrow k+1$. \\
            Compute matrices $\mathbf{B}$ and $\mathbf{C}$ following \eqref{eq:KKT_cond}. \\
            Update $\mathbf{f}^{(k)} \leftarrow \mathbf{C}^{\dagger}\mathbf{B} \mathbf{f}^{(k-1)} / \|\mathbf{C}^{\dagger}\mathbf{B} \mathbf{f}^{(k-1)}\|_2$.
        }
        
        Set $n$-th basis vector: $\mathbf{w}_n \leftarrow \mathbf{f}^{(k)}$.

        Update projection matrix: $\mathbf{\Pi} \leftarrow \mathbf{\Pi} - \mathbf{w}_n\mathbf{w}_n^H$.
        
        Update projected effective gain matrices $\tilde{\mathbf{G}}^{(c)}, \tilde{\mathbf{G}}^{(s)}, \tilde{\mathbf{G}}^{(e)}$ using Proposition~\ref{prop:greedy_update}. \\
    }
    \Return{$\mathbf{W} = [\mathbf{w}_1, \dots, \mathbf{w}_{N_s}]$}
\end{algorithm}

\subsection{Stage 2: Power Allocation}
After Stage 1 yields a basis matrix $\mathbf{W}$, Stage 2 optimizes 
the power allocation vector $\mathbf{p} = [p_1, \dots, p_{N_s}]^T$ 
over the corresponding fixed basis vectors:
\begin{align}
    \max_{\{p_k\}} \; & w_c \log\det(\mathbf{I} \hspace{-0.05em}+\hspace{-0.05em} \mathbf{K}_c\diag(\mathbf{p})) \hspace{-0.07em}-\hspace{-0.07em} w_c \log\det(\mathbf{I} \hspace{-0.05em}+\hspace{-0.05em} \mathbf{K}_e\diag(\mathbf{p})) 
     \nonumber \\
    & + w_s \log\det(\mathbf{I} \hspace{-0.05em}+\hspace{-0.05em} \mathbf{K}_s\diag(\mathbf{p}))\label{eq:power_prob} \\
    \text{s.t.} \quad & \mathbf{p} \ge 0, \quad \mathbf{1}^T\mathbf{p} \le P_{\mathrm{tot}}. \nonumber
\end{align}
where $\mathbf{K}_i = \mathbf{W}^H \mathbf{H}_i^H \mathbf{H}_i \mathbf{W}$ for $i \in \{c, s, e\}$. 
The objective function is composed of two concave terms and one convex term in $\mathbf{p}$ (the negative eavesdropper's MI). This structure makes the overall problem a nonconvex DC program.

To solve \eqref{eq:power_prob}, we employ SCA for the DC objective. 
At each iteration, we linearize the challenging convex part of the objective using its first-order Taylor expansion around the power allocation $\mathbf{p}^{(k-1)}$ from the previous iteration. 
Substituting this linear approximation into the objective function and discarding constants independent of $\mathbf{p}$ yields a convex subproblem at iteration $k$. 
This subproblem admits a concave surrogate objective that globally lower-bounds the original objective and is tight at $\mathbf{p}^{(k-1)}$. 
It can be solved efficiently via projected gradient ascent. Iteratively maximizing this lower bound yields a monotonic improvement of the objective value and converges to a stationary point under standard regularity conditions.
Similar to Stage 1, we employ a warm-start strategy for the initialization of the SCA procedure. Specifically, instead of resetting to a uniform power allocation at each outer iteration, we set the initial point for the current stage, $\mathbf{p}^{(0)}$, to the converged power vector obtained in the previous outer iteration. The resulting procedure is detailed in Algorithm \ref{alg:power_allocation}.

The overall precoder design process is carried out as follows. 
First, Stage 1 is executed to update the precoding basis $\mathbf{W}$ for the given power. Then, using this new basis, Stage 2 is performed to re-optimize the power allocation $\mathbf{P}$. 
This two-stage cycle is repeated until the weighted sum rate in \eqref{eq:prob-log} converges, ensuring a joint optimization of both the precoding basis and the power allocation.

\begin{algorithm}[t]
    \caption{Power Allocation for Fixed Basis}
    \label{alg:power_allocation}
    \KwInput{Channel matrices $\mathbf{H}_c, \mathbf{H}_s, \mathbf{H}_e$; basis matrix $\mathbf{W}$; weights $w_c, w_s$; total power $P_{\mathrm{tot}}$; tolerance $\varepsilon$; previous power vector $\mathbf{p}^{\mathrm{prev}}$}
    \KwOutput{Power allocation vector $\mathbf{p}$}
    
    \text{Initialize } $\mathbf{p}^{(0)} \leftarrow \mathbf{p}^{\mathrm{prev}}$ and $k \leftarrow 0$. \\
    \text{Compute } $\mathbf{K}_i \leftarrow \mathbf{W}^{H} \mathbf{H}_i^{H} \mathbf{H}_i \mathbf{W}$ for $i \in \{c, s, e\}$.
    
    \Repeat{$\|\mathbf{p}^{(k)} - \mathbf{p}^{(k-1)}\|_2 < \varepsilon$}{
        $k \leftarrow k+1$ \\
        Compute the gradient: $\mathbf{k}_e^{(k-1)} \leftarrow \frac{1}{\ln 2}\diag \big( (\mathbf{I} + \mathbf{K}_e \diag(\mathbf{p}^{(k-1)}))^{-1} \mathbf{K}_e \big)$ \\
        
        Solve the convex subproblem to obtain $\mathbf{p}^{(k)}$: \\
        $\displaystyle \mathbf{p}^{(k)} \leftarrow \argmax_{\mathbf{p} \ge \mathbf{0},\, \mathbf{1}^{T}\mathbf{p} \leq P_{\mathrm{tot}}}  w_c \log_2\det(\mathbf{I} + \mathbf{K}_c \diag(\mathbf{p})) $ \\
        $\displaystyle + w_s \log_2\det(\mathbf{I} + \mathbf{K}_s \diag(\mathbf{p})) - w_c \, (\mathbf{k}_e^{(k-1)})^{T} \mathbf{p}$
    }
    \Return{$\mathbf{p}^{(k)}$}
\end{algorithm}

\subsection{Discussions} \label{sec:asymptotic_optimality}
The proposed precoding method appropriately incorporates the optimal precoder structure analyzed in the previous section. 
At low SNR, the basis construction stage would find the principal eigenvector of the composite matrix $\mathbf{M} = w_c \mathbf{H}_c^H \mathbf{H}_c - w_c \mathbf{H}_e^H \mathbf{H}_e + w_s \mathbf{H}_s^H \mathbf{H}_s$ as the first basis vector. 
Subsequently, the power allocation stage correctly allocates the entire power budget to this single stream, ensuring the overall algorithm converges to the globally optimal rank-one precoder. 
In the high-SNR regime, the allocated power $p_n$ for any useful stream is large, causing the matrices $\mathbf{A}^{(i)}_{n-1}$ in our objective to be dominated by the effective channel gain matrix, i.e., $\mathbf{A}^{(i)}_{n-1} \approx p_n \mathbf{G}^{(i)}_{n-1}$. 
In this case, the fixed-point iteration in \eqref{eq:iter_matrix} seeks a direction $\mathbf{f}$ that maximizes the gains from the communication and sensing channels while minimizing leakage to the eavesdropper. Specifically, the matrix $\mathbf{B}(\mathbf{f})$ combines the effective channel gain matrices $\mathbf{G}^{(c)}$ and $\mathbf{G}^{(s)}$, thereby amplifying vector components within the communication and sensing row spaces.
In contrast, the use of $(\mathbf{C}(\mathbf{f}))^{\dagger}$ effectively penalizes directions that are strong in the eavesdropper's channel, thereby promoting solutions that lie within or near its null space. This iterative process serves as a numerical method for finding directions with the highest directional DoF weight. 
This is designed to sequentially populate the basis vectors that span $\CMcal{V}_{\text{useful}}$. 
These observations suggest that the proposed algorithm is well-founded and capable of recovering optimal solutions in key asymptotic regimes. 
Additionally, this two-stage architecture is designed to effectively embody the structural insights derived from our analysis. By explicitly separating the spatial basis construction from the power distribution, the algorithm can focus on identifying the most beneficial signal directions—those that balance communication and sensing gains against eavesdropping leakage—without being hindered by the coupling with power variables. Moreover, this decoupling transforms the challenging joint optimization problem into a sequence of computationally efficient subproblems. This ensures high scalability, making the proposed design well-suited for practical implementation in large-scale antenna systems.

\section{Simulation Results}\label{sec:simulations}
In this section, we present numerical results evaluating the proposed two-stage precoder design for the MIMO-ME-MS channel. We first illustrate the secrecy–sensing tradeoff region (Figs.~\ref{fig:pareto_0db}–\ref{fig:pareto_20db}), then plot the weighted sum rate versus SNR (Fig.~\ref{fig:sumrate_snr}), and finally compare computational complexity (Fig.~\ref{fig:cpu_time}).

\subsection{Simulation Setup and Baselines}
We consider Rayleigh fading channels with entries distributed as $\mathcal{CN}(0,1)$, and we normalize the noise variances so that the SNR is controlled by $P_{\mathrm{tot}}$. For the achievable-region experiments (Figs.~\ref{fig:pareto_0db}–\ref{fig:pareto_20db}), we set $n_t=n_c=n_e=n_s=16$ and average the results over multiple independent channel realizations.
We set the number of data streams to $N_s=2$ at low SNR (0 dB) and $N_s=12$ at high SNR (20 dB), 
following standard stream allocation practices in MIMO systems.
To trace the secrecy--sensing Pareto boundary, we sweep $w_c\in[0,1]$ with $w_s=1-w_c$.

We compare the proposed two-stage algorithm with the following baselines: 
\begin{itemize}
  \item \textbf{WMMSE-based ISAC precoding \cite{wangUnifiedISACPareto2024}}:
  This baseline solves the MIMO-MS weighted sum rate maximization problem based on the communication MI and the SMI (i.e., the communication--sensing tradeoff without an eavesdropper) using a WMMSE-based procedure. We then evaluate the secrecy rate and SMI achieved by the resulting precoder.

  \item \textbf{GSVD-based secrecy precoding \cite{KhistiMIMOME2010}}:
  This baseline designs the precoder for the MIMO-ME wiretap channel using the classical GSVD-based structure. Since it does not account for sensing, we evaluate its SMI under the resulting secrecy-oriented precoder.

  \item \textbf{SCA-SDR-based precoding}: 
    To the best of our knowledge, no existing baseline directly addresses precoder optimization under the considered MIMO-ME-MS setup. 
    For this reason, as a baseline aligned with our joint objective, we construct an SCA-SDR-based precoding by adapting and modifying the formulation in \cite{li:tvt:25}. Specifically, since the secrecy rate induces a DC structure in the objective, we employ the SCA technique to linearize the nonconvex terms. By linearizing the eavesdropper log-det term using the Taylor approximation and relaxing the rank constraint (i.e., SDR), the problem is transformed into a sequence of convex SDP subproblems. Each SDP subproblem is solved using a generic interior-point solver (via CVX). 
    Upon convergence, we recover a rank-$N_s$ precoder from the optimized covariance using eigen-decomposition.
    \end{itemize}

We also include (i) a sensing upper bound (SUB) obtained by maximizing SMI subject to the power constraint, which corresponds to the sensing-only operating point, and (ii) a time-sharing baseline constructed by the convex combination of the secrecy-only (GSVD) and sensing-only (SUB) designs.

\subsection{Numerical Results}

\begin{figure}[t]
  \centering
  \includegraphics[width=\columnwidth]{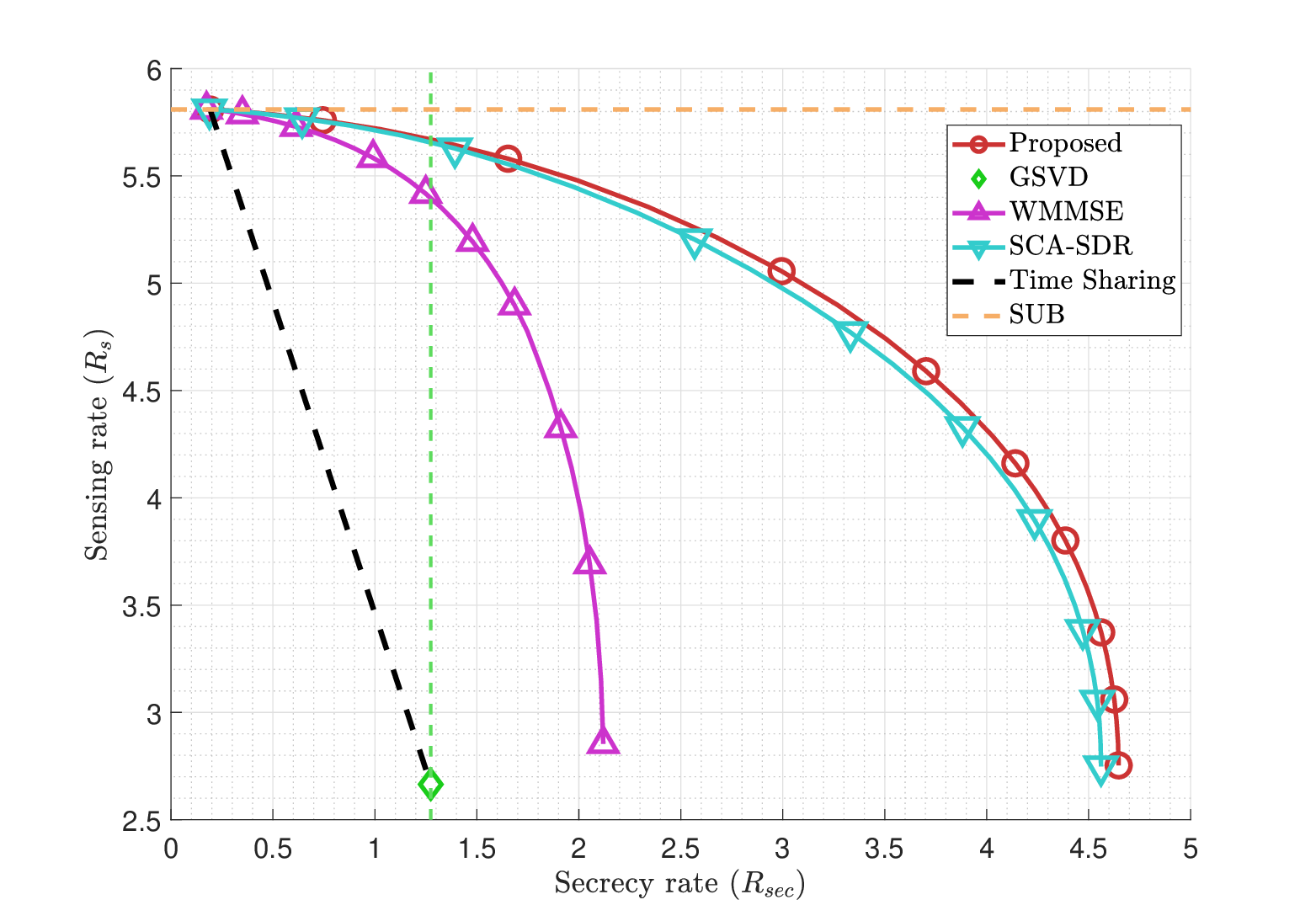}
  \caption{Achievable region of $(R_{\mathrm{sec}},\, R_s)$ at 0 dB SNR.}
  \label{fig:pareto_0db}
\end{figure}

\begin{figure}[t]
  \centering
  \includegraphics[width=\columnwidth]{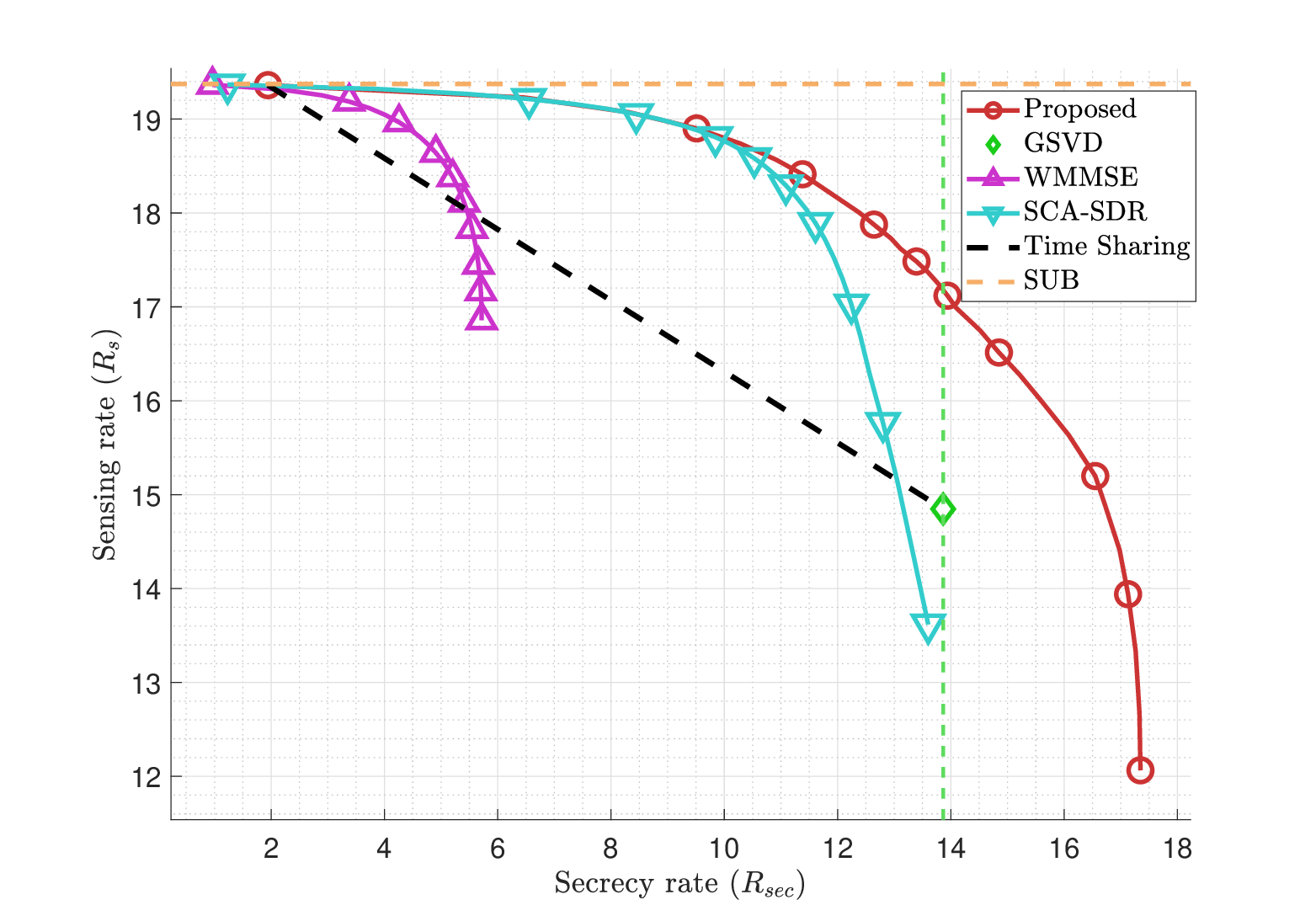}
  \caption{Achievable region of $(R_{\mathrm{sec}},\, R_s)$ at 20 dB SNR.}
  \label{fig:pareto_20db}
\end{figure}

Figs.~\ref{fig:pareto_0db} and \ref{fig:pareto_20db} illustrate the achievable Pareto boundaries between the secrecy rate $R_{\mathrm{sec}}$ and the SMI $R_s$ at low (0~dB) and high (20~dB) SNR, respectively.

In Fig.~\ref{fig:pareto_0db} (0 dB SNR), several key observations can be made. 
First, the proposed method significantly outperforms the naive time-sharing baseline, demonstrating the substantial gains from the joint optimization of secrecy and sensing objectives. 
Second, the WMMSE-based precoding, while achieving a strong communication--sensing tradeoff, cannot incorporate the eavesdropper and consequently achieves a lower secrecy rate than the proposed method for any given level of SMI. 
The GSVD-based precoding, designed solely to maximize the secrecy rate, does not consider sensing, and its performance is therefore independent of the system weights; it is thus represented by a single operating point. Most notably, at this low SNR, this GSVD-based point is suboptimal even in terms of secrecy rate. 
Compared to the SCA-SDR precoding, which jointly accounts for the secrecy rate and the SMI, the proposed method still achieves better performance. 
The rationale behind these performance gains is discussed below.

Fig.~\ref{fig:pareto_20db} presents the same tradeoff at 20 dB SNR. 
The Pareto boundary traced by the proposed method again demonstrates significant gains over the naive time-sharing baseline. 
While the secrecy rate achieved by the GSVD-based precoding improves relative to the low-SNR case, it remains suboptimal even at 20 dB. This highlights that for the GSVD-based precoding to approach its theoretical asymptotic optimality for the secrecy-only objective, an SNR regime far greater than 20 dB would be required, which is often impractical. 
In contrast, the proposed method demonstrates robust and superior performance in this practical high-SNR regime.
Furthermore, the limitation of the secrecy-agnostic WMMSE-based precoding becomes more pronounced in this high-SNR regime. Since this method maximizes the sum of communication MI and the SMI without penalizing the leakage to the eavesdropper, the unsuppressed signal power received by the eavesdropper grows significantly with SNR. Consequently, neglecting the eavesdropper term leads to a much more severe degradation in secrecy rate at 20 dB compared to the low-SNR case, highlighting the necessity of the proposed joint security-aware design.
Additionally, the SCA-SDR-based precoding exhibits a substantial performance gap relative to the proposed method at 20 dB, particularly in the high-secrecy region. 
These performance gains over the SCA-SDR-based precoding stem from the following factors. The proposed method sequentially finds the precoding basis by reflecting the useful subspaces identified in our analysis. 
Consequently, the resulting transmit subspace remains structurally intact without resorting to rank reduction. 
By contrast, the SCA-SDR-based precoding does not explicitly leverage such subspace-level insights and instead relies on SDR followed by rank-$N_s$ extraction. 
This can lead to performance degradation.
This confirms that the proposed method is particularly beneficial for achieving the maximum secrecy DoF at high SNR.
In summary, the proposed method consistently outperforms the baselines across both low- and high-SNR regimes, demonstrating effective operation over a wide range of SNR conditions.

Fig.~\ref{fig:sumrate_snr} plots the weighted sum rate ($w_c R_{\mathrm{sec}} + w_s R_s$ with $w_c=w_s=0.5$) versus SNR for different antenna configurations ($n_t=n_c=n_e=n_s \in \{16, 32, 64\}$), where the number of data streams is set to $N_s = n_t/2$. The SCA-SDR-based precoding was excluded from this simulation due to its high computational complexity associated with the $n_t \times n_t$ SDP variable, which becomes prohibitive as the number of antennas increases (see Fig.~\ref{fig:cpu_time}).

\begin{figure}[t]
  \centering
  \includegraphics[width=\columnwidth]{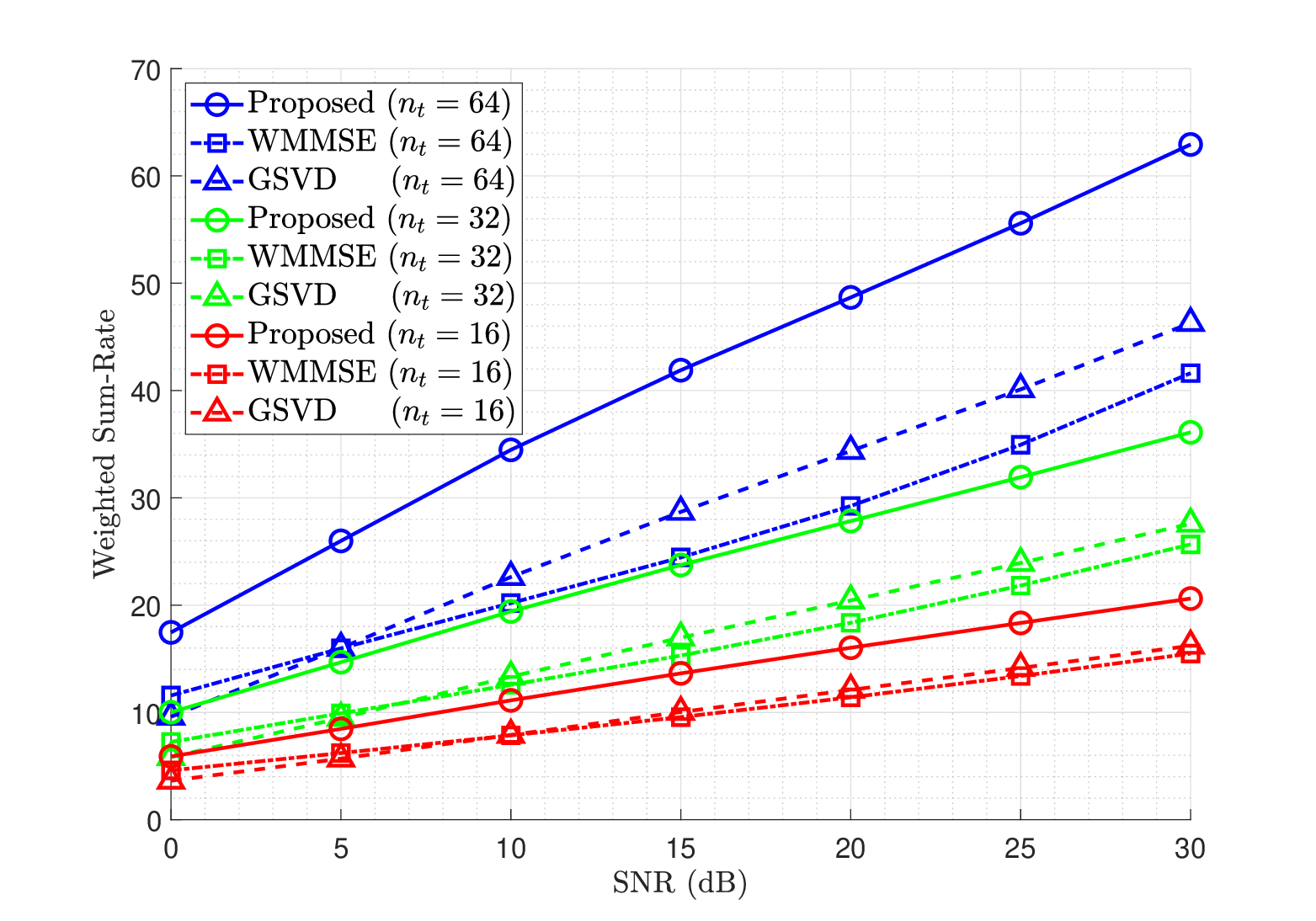}
  \caption{Weighted sum rate $w_c R_{\mathrm{sec}} + w_s R_s$ versus SNR for different numbers of antennas ($w_c=w_s=0.5$).}
  \label{fig:sumrate_snr}
\end{figure}

The results clearly show that for all antenna configurations, the proposed method consistently outperforms the other baseline schemes across the entire SNR range. 
Notably, at high SNR, the WMMSE-based precoding performs even worse than the sensing-agnostic GSVD-based precoding. This indicates that without explicit eavesdropper nulling, the severe leakage penalty outweighs the multiplexing gains from the large array, rendering the secrecy-agnostic design ineffective.
Furthermore, the performance gap between the proposed method and the baselines widens as the number of antennas increases. 
This is because a larger antenna array provides a higher-dimensional transmit space, requiring a more delicate consideration of the subspaces for both secrecy and sensing---a task for which the GSVD- and WMMSE-based precoders are ill-equipped, but which the proposed method is explicitly designed to handle. 
This highlights the increasing importance of the proposed joint design approach in systems with larger antenna arrays.

\subsection{Complexity Analysis} \label{sec:complexity}
We characterize the computational complexity of the proposed method and the SCA-SDR-based precoding in big-$O$ notation. Additionally, we corroborate the scaling trends via average CPU execution time measurements.

The proposed two-stage algorithm decouples the basis construction from the power allocation, significantly reducing the computational burden.
    Let $I_{\mathrm{out}}$ denote the number of outer iterations, $I_{\mathrm{fp}}$ the average fixed-point iterations per basis vector in Stage 1, and $I_{\mathrm{pa}}$ the iterations for the power allocation subproblem in Stage 2. 
In Stage 1, constructing $N_s$ basis vectors is dominated by matrix operations of order $O(n_t^3)$, yielding $O(I_{\mathrm{fp}} N_s n_t^3)$. 
In Stage 2, the power allocation involves $N_s$ scalar variables with complexity $O(I_{\mathrm{pa}} N_s^{3.5})$. 
Crucially, owing to the warm-start strategy in which the solver is initialized with the solution from the previous iteration, these iteration counts ($I_{\mathrm{out}}, I_{\mathrm{fp}}, I_{\mathrm{pa}}$) remain small in practice.
Consequently, the total complexity scales as $O(I_{\mathrm{out}}(I_{\mathrm{fp}} N_s n_t^3 + I_{\mathrm{pa}} N_s^{3.5}))$, which is a low-order polynomial in $n_t$. 
In contrast, the SCA-SDR-based precoding \cite{li:tvt:25} requires solving a semidefinite program with an $n_t \times n_t$ matrix variable at each of its $I_{\mathrm{sdp}}$ iterations. This incurs a per-iteration complexity of $O(n_t^{6.5})$ \cite{luo:spm:10}, leading to a total complexity of $O(I_{\mathrm{sdp}} n_t^{6.5})$.
This prohibitive polynomial scaling with respect to $n_t$ makes the SDP-based approach computationally infeasible for large antenna arrays, whereas the proposed method maintains scalability.

\begin{figure}[!t]
  \centering
  \includegraphics[width=\columnwidth]{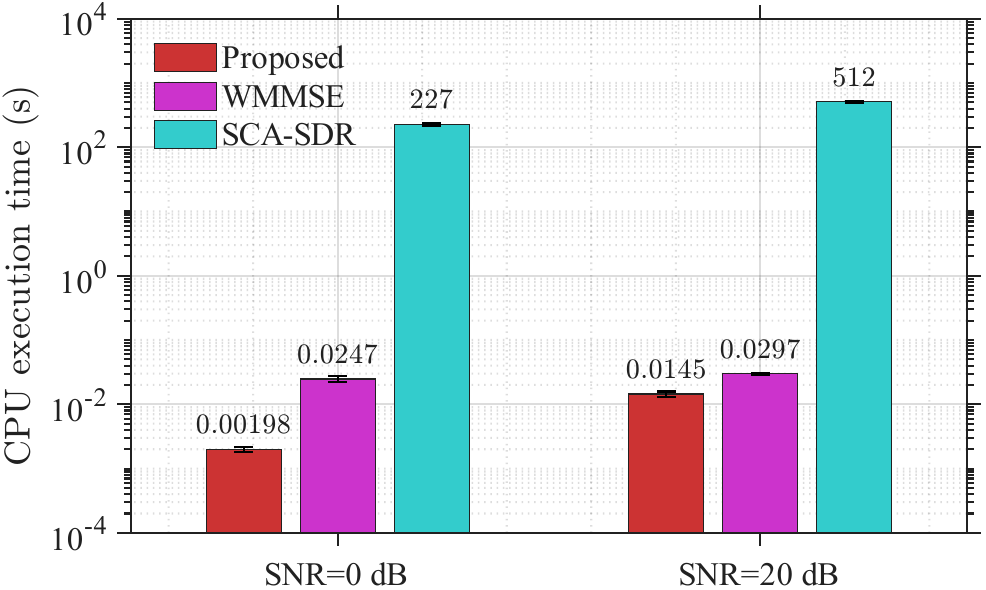}
  \caption{Average CPU execution time at 0 dB and 20 dB SNR ($n_t=16$).}
  \label{fig:cpu_time}
\end{figure}

To empirically verify the computational complexity analysis, we measure the average CPU execution time of the algorithms. 
It is worth noting that while CPU time is implementation-dependent and not an absolute measure of complexity, it serves as a useful indirect metric for gauging the relative computational burden of different methods. 
Fig.~\ref{fig:cpu_time} reports the average execution time required to obtain a single solution point corresponding to the Pareto boundaries shown in Figs.~\ref{fig:pareto_0db} and \ref{fig:pareto_20db} (where $n_t=16$).

The results reveal a dramatic contrast in runtime efficiency. 
At 0 dB SNR, the SCA-SDR-based precoding requires approximately 227 seconds to converge, whereas the proposed method completes the optimization in merely 1.98 ms. 
This implies that the SCA-SDR-based precoding is on the order of $10^5$ times slower than the proposed method, i.e., the proposed method requires only about $0.001\%$ of the runtime of the SCA-SDR-based precoding.
Similarly, at 20 dB SNR, while the proposed method takes slightly longer (14.5 ms) due to the increased number of streams, the SCA-SDR-based precoding takes over 512 seconds, maintaining a speed gap of more than four orders of magnitude. 
This indicates that the proposed method requires only about $0.003\%$ of the runtime of the SCA-SDR-based precoding.
The prohibitive computational cost of the SCA-SDR approach stems from the need to solve high-dimensional SDP subproblems with lifted variables at every iteration. 
In contrast, the proposed method maintains extremely low computational cost by relying on efficient matrix-vector operations, confirming its high scalability and suitability for practical implementation in secure ISAC systems.

\section{Conclusion} \label{sec:conclusion}
In this work, we investigated the fundamental performance limits of secure ISAC by introducing and analyzing the MIMO-ME-MS channel. 
By adopting a unified information-theoretic framework based on SMI, we formulated the joint design of the secure communication and sensing precoder as a weighted sum rate maximization problem. 
This formulation captures the intrinsic tradeoffs among the three competing objectives. 
Our primary theoretical contribution is a comprehensive high-SNR analysis based on a fundamental decomposition of the transmit space into eight subspaces. 
This analysis allowed us to explicitly characterize the maximum achievable weighted DoF and identify the structure of a quasi-optimal precoder. 
A key insight is that the optimal precoder must exclusively span ``useful subspaces'' composed of all DoF-positive dimensions, the composition of which depends critically on the system weights assigned to secrecy and sensing. 
This characterization revealed the inadequacy of directly extending known schemes from the simpler MIMO-ME or MIMO-MS subproblems. 

Building on these structural insights, we proposed a practical precoding method. To address the nonconvex nature of the precoder design problem, we developed a principled two-stage algorithm. 
The algorithm alternates between a basis construction stage that sequentially identifies orthogonal vectors to maximize marginal rate gain and a power allocation stage. The power allocation stage solves the resulting DC program via SCA. 
The numerical simulations demonstrated that the proposed precoder achieves substantial gains in the MIMO-ME-MS channel. 
These gains stem from its capability to strike a balance among the conflicting objectives associated with communication, secrecy, and sensing. 
Through these contributions, we established a theoretical and algorithmic foundation for the MIMO-ME-MS channel. 

This work also suggests several directions for future research. 
First, while we assumed perfect knowledge of the involved channels to obtain a clean high-SNR characterization, practical systems typically operate with imperfect channel estimates. In particular, the eavesdropper's CSI is often uncertain or even unavailable \cite{ren:tcom:23}. 
Under channel uncertainty, the transmit-side row/null spaces that underpin our derivations become perturbed, potentially inducing power leakage into the eavesdropper's effective subspace. 
In the high-SNR regime, such leakage may severely degrade the secrecy DoF unless the estimation error decays sufficiently rapidly with SNR.
A natural direction is robust secure-ISAC precoder design, for example, maximizing a worst-case weighted-SMI objective under bounded or stochastic channel errors. 
Second, we focused on the canonical single-RX/single-eavesdropper/single-target model to isolate the fundamental three-way interaction and to keep the DoF characterization tractable. 
Extending the framework to multiple legitimate users, multiple sensing targets, or multiple eavesdroppers \cite{hiaTcom22} introduces additional challenges due to multi-user interference and diverse sensing criteria. 
Developing scalable algorithms and structural characterizations for these generalizations, while also accounting for hardware constraints such as hybrid beamforming, remains an important direction for future work. 
Finally, it is also interesting to incorporate finite alphabet inputs into the MIMO-ME-MS channel \cite{wu:tcom:17, jin:tsp:17}. This constraint fundamentally changes the structure of the MI expressions from the log-det form, introducing new challenges in characterizing the optimal precoder structure.

\appendices
\section{Proof of Theorem~\ref{thm:mems_decomp}} \label{sec:appen1}

\begin{lemma}\label{lem:subspace_decomp}
Let $\CMcal{U}$ be a subspace of a finite-dimensional inner-product space, and let $\CMcal{A}$ be a subspace of $\CMcal{U}$.
Then,
\begin{align}
    \CMcal{U} = \CMcal{A} \oplus \left(\CMcal{U} \cap \CMcal{A}^\perp\right).
    \label{eq:lemma1}
\end{align}
\end{lemma}

\begin{IEEEproof}
    For any $\mathbf{u}\in\CMcal{U}$, the projection theorem yields the orthogonal decomposition $\mathbf{u}=\mathcal{P}_{\CMcal{A}}\mathbf{u}+\mathcal{P}_{\CMcal{A}^\perp}\mathbf{u}$, where the operator $\mathcal{P}_{\CMcal{S}}$ denotes the orthogonal projection onto the subspace $\CMcal{S}$. 
    Since $\mathcal{P}_{\CMcal{A}}\mathbf{u}\in\CMcal{A}\subseteq\CMcal{U}$ and $\mathbf{u}\in\CMcal{U}$, we have
    $\mathcal{P}_{\CMcal{A}^\perp}\mathbf{u}=\mathbf{u}-\mathcal{P}_{\CMcal{A}}\mathbf{u}\in\CMcal{U}$.
    Moreover, $\mathcal{P}_{\CMcal{A}^\perp}\mathbf{u}\in\CMcal{A}^\perp$, hence
    $\mathcal{P}_{\CMcal{A}^\perp}\mathbf{u}\in\CMcal{U}\cap\CMcal{A}^\perp$.
    Orthogonality and uniqueness follow from $\CMcal{A}\perp(\CMcal{U}\cap\CMcal{A}^\perp)$ and finite dimensionality.
\end{IEEEproof}
Now, we prove Theorem~\ref{thm:mems_decomp}.

\begin{IEEEproof}[Proof of Theorem~\ref{thm:mems_decomp}]
\textbf{Step 1} {(Decomposition of the null space $\CMcal{N}_e$)}:
We first focus on the null space of the eavesdropper, $\CMcal{N}_e$. Applying Lemma~\ref{lem:subspace_decomp} with $\CMcal{U} = \CMcal{N}_c \cap \CMcal{N}_e$ and $\CMcal{A} = \CMcal{V}_n$ yields
\begin{align}
    \CMcal{N}_c \cap \CMcal{N}_e &= \CMcal{V}_n \oplus \left(\CMcal{N}_c \cap \CMcal{N}_e \cap \CMcal{V}_n^\perp\right) = \CMcal{V}_n \oplus \CMcal{V}_s.
\end{align}
Similarly, applying Lemma~\ref{lem:subspace_decomp} with $\CMcal{U} = \CMcal{N}_s \cap \CMcal{N}_e$ and $\CMcal{A} = \CMcal{V}_n$ gives
\begin{align}
    \CMcal{N}_s \cap \CMcal{N}_e &= \CMcal{V}_n \oplus \left(\CMcal{N}_s \cap \CMcal{N}_e \cap \CMcal{V}_n^\perp\right) = \CMcal{V}_n \oplus \CMcal{V}_c.
\end{align}
To verify that $\CMcal{V}_n \oplus \CMcal{V}_c \oplus \CMcal{V}_s$ forms a direct sum, assume $\mathbf{x}_n + \mathbf{x}_c + \mathbf{x}_s = \mathbf{0}$ with $\mathbf{x}_j \in \CMcal{V}_j$. Since $\mathbf{x}_c, \mathbf{x}_s \in \CMcal{V}_n^\perp$, we have $\mathbf{x}_n = -(\mathbf{x}_c + \mathbf{x}_s) \in \CMcal{V}_n^\perp$, implying $\mathbf{x}_n = \mathbf{0}$. Then $\mathbf{x}_c = -\mathbf{x}_s$. Since $\mathbf{x}_c \in \CMcal{N}_s \cap \CMcal{N}_e$ and $-\mathbf{x}_s \in \CMcal{N}_c \cap \CMcal{N}_e$, it follows that $\mathbf{x}_c \in \CMcal{N}_c \cap \CMcal{N}_s \cap \CMcal{N}_e = \CMcal{V}_n$. But $\mathbf{x}_c \in \CMcal{V}_n^\perp$, so $\mathbf{x}_c = \mathbf{0}$, and consequently $\mathbf{x}_s = \mathbf{0}$.

Now, applying Lemma~\ref{lem:subspace_decomp} to the entire $\CMcal{N}_e$ with $\CMcal{A} = \CMcal{V}_n \oplus \CMcal{V}_c \oplus \CMcal{V}_s$, we obtain
\begin{align}
    \CMcal{N}_e = (\CMcal{V}_n \oplus \CMcal{V}_c \oplus \CMcal{V}_s) \oplus \left(\CMcal{N}_e \cap (\CMcal{V}_n \oplus \CMcal{V}_c \oplus \CMcal{V}_s)^\perp\right).
\end{align}
By the definition in Table~\ref{tab:subspace_dof}, the second term is precisely $\CMcal{V}_{cs}$, yielding the direct sum
\begin{align}
    \CMcal{N}_e = \CMcal{V}_n \oplus \CMcal{V}_c \oplus \CMcal{V}_s \oplus \CMcal{V}_{cs}. \label{eq:Ne_decomp}
\end{align}

\textbf{Step 2} {(Decomposition of the row space $\CMcal{R}_e$)}:
Next, we decompose the row space $\CMcal{R}_e$ using an identical structural approach. Applying Lemma~\ref{lem:subspace_decomp} with $\CMcal{U} = \CMcal{R}_c \cap \CMcal{R}_e$ and $\CMcal{A} = \CMcal{V}_{cse}$ gives
\begin{align}
    \CMcal{R}_c \cap \CMcal{R}_e &= \CMcal{V}_{cse} \oplus \left(\CMcal{R}_c \cap \CMcal{R}_e \cap \CMcal{V}_{cse}^\perp\right) = \CMcal{V}_{cse} \oplus \CMcal{V}_{ce}.
\end{align}
Similarly, for $\CMcal{U} = \CMcal{R}_s \cap \CMcal{R}_e$ and $\CMcal{A} = \CMcal{V}_{cse}$, we have
\begin{align}
    \CMcal{R}_s \cap \CMcal{R}_e &= \CMcal{V}_{cse} \oplus \left(\CMcal{R}_s \cap \CMcal{R}_e \cap \CMcal{V}_{cse}^\perp\right) = \CMcal{V}_{cse} \oplus \CMcal{V}_{se}.
\end{align}
To verify the direct sum $\CMcal{V}_{cse} \oplus \CMcal{V}_{ce} \oplus \CMcal{V}_{se}$, assume $\mathbf{x}_{cse} + \mathbf{x}_{ce} + \mathbf{x}_{se} = \mathbf{0}$. Since $\mathbf{x}_{ce}, \mathbf{x}_{se} \in \CMcal{V}_{cse}^\perp$, we get $\mathbf{x}_{cse} = -(\mathbf{x}_{ce} + \mathbf{x}_{se}) \in \CMcal{V}_{cse}^\perp$, so $\mathbf{x}_{cse} = \mathbf{0}$. Then $\mathbf{x}_{ce} = -\mathbf{x}_{se}$. This implies $\mathbf{x}_{ce} \in (\CMcal{R}_c \cap \CMcal{R}_e) \cap (\CMcal{R}_s \cap \CMcal{R}_e) = \CMcal{V}_{cse}$. Because $\mathbf{x}_{ce} \in \CMcal{V}_{cse}^\perp$, we conclude $\mathbf{x}_{ce} = \mathbf{0}$ and $\mathbf{x}_{se} = \mathbf{0}$.

Applying Lemma~\ref{lem:subspace_decomp} to the entire $\CMcal{R}_e$ with $\CMcal{A} = \CMcal{V}_{cse} \oplus \CMcal{V}_{ce} \oplus \CMcal{V}_{se}$ yields
\begin{align}
    \CMcal{R}_e = (\CMcal{V}_{cse} \oplus \CMcal{V}_{ce} \oplus \CMcal{V}_{se}) \oplus \left(\CMcal{R}_e \cap (\CMcal{V}_{cse} \oplus \CMcal{V}_{ce} \oplus \CMcal{V}_{se})^\perp\right).
\end{align}
The second term is exactly $\CMcal{V}_e$, which provides the direct sum
\begin{align}
    \CMcal{R}_e = \CMcal{V}_{cse} \oplus \CMcal{V}_{ce} \oplus \CMcal{V}_{se} \oplus \CMcal{V}_e. \label{eq:Re_decomp}
\end{align}

\textbf{Step 3} {(Combination of null and row spaces)}:
Since $\CMcal{R}_e$ is the orthogonal complement of $\CMcal{N}_e$, the entire vector space can be written as the orthogonal direct sum $\CMcal{N}_e \oplus \CMcal{R}_e$. Substituting the decompositions from \eqref{eq:Ne_decomp} and \eqref{eq:Re_decomp} yields the complete 8-subspace direct sum 
\begin{align}
    \mathbb{C}^{n_t} = \CMcal{N}_e \oplus \CMcal{R}_e = \bigoplus_{j\in\mathcal{K}} \CMcal{V}_j, \label{eq:full_decomp}
\end{align}
which proves \eqref{eq:mems_spaces}. 
\end{IEEEproof}

\section{Proof of Theorem~\ref{thm:dof_max}} \label{sec:appen2}

\begin{lemma} \label{lem:2}
The weighted DoF of a quasi-optimal precoder, $d(\mathbf{F}_{\text{q-opt}})$, is upper-bounded by the maximum value of the rank-based expression, assuming fixed precoding basis $\mathbf{W}$:
\begin{align}
\hspace{-0.2em}d(\mathbf{F}_{\text{q-opt}})\hspace{-0.1em} \le \hspace{-0.1em}\max_{\mathbf{F}} \hspace{-0.1em}\left( w_c \rank(\mathbf{H}_c\mathbf{F}) \hspace{-0.1em}-\hspace{-0.1em} w_c \rank(\mathbf{H}_e\mathbf{F}) \hspace{-0.1em}+\hspace{-0.1em} w_s \rank(\mathbf{H}_s\mathbf{F}) \right).
\end{align}
\end{lemma}

\begin{IEEEproof}
To find an upper bound on the weighted DoF, we relax the total power constraint and analyze the optimal power allocation for each column of a precoder $\mathbf{F}$.

First, we show that a DoF-optimal power profile must be binary. As established in Proposition~\ref{prop:greedy} with fixed basis $\mathbf{W}$, the marginal rate gain from the $n$-th column with power $p_n$ is a sum of logarithmic terms of the form $\log(1+p_n \mathbf{w}_n^H \mathbf{G}_{n-1} \mathbf{w}_n)$. The DoF contribution from this column is therefore linear with respect to its power scaling exponent $\alpha$ (where $p_n \sim (P_{\mathrm{tot}})^\alpha$), as the effective gain term $\mathbf{w}_n^H \mathbf{G}_{n-1} \mathbf{w}_n$ is independent of $p_n$. This linearity implies that any intermediate power scaling ($0 < \alpha < 1$) is suboptimal for DoF maximization. Thus, each column's power must scale as either $O(P_{\mathrm{tot}})$ (for $\alpha=1$) or as a constant (for $\alpha=0$). Since a constant power allocation yields zero DoF, it is equivalent to zero power from a DoF perspective.

Second, based on the above, we only need to consider precoders where each column is allocated either $O(P_{\mathrm{tot}})$ power or zero power to maximize the DoF. For any such precoder $\mathbf{F}$, its weighted DoF is precisely given by:
\begin{align}
    d(\mathbf{F}) = w_c \rank(\mathbf{H}_c\mathbf{F}) - w_c \rank(\mathbf{H}_e\mathbf{F}) + w_s \rank(\mathbf{H}_s\mathbf{F}). \label{eq:dof_formula_appendix}
\end{align}
The weighted DoF of a quasi-optimal precoder, $d(\mathbf{F}_{\text{q-opt}})$, must be equal to the value of \eqref{eq:dof_formula_appendix} for some specific choice of $\mathbf{F}$. This value is necessarily less than or equal to the maximum possible value of the expression over all choices of $\mathbf{F}$.
\end{IEEEproof}

Now, we prove the main theorem.

\begin{IEEEproof}[Proof of Theorem~\ref{thm:dof_max}]
For any precoder $\mathbf{F}$ and positive weights $w_c, w_s > 0$, we show that the weighted sum of ranks is upper-bounded by:
\begin{align} \label{prop:weighted_inequality}
& w_c \rank(\mathbf{H}_c \mathbf{F}) - w_c \rank(\mathbf{H}_e \mathbf{F}) + w_s \rank(\mathbf{H}_s \mathbf{F}) \nonumber \\ 
&\le w_c \rank(\mathbf{F}_c) + w_s \rank(\mathbf{F}_s) + (w_c+w_s)\rank(\mathbf{F}_{cs})\nonumber \\
&\quad + w_s \rank(\mathbf{F}_{cse}) + w_s \rank(\mathbf{F}_{se}) \nonumber \\
&\quad - \min\{w_c,w_s\}[\rank(\mathbf{F}_{se})-\rank(\mathbf{F}_{ce})]^+.
\end{align}

We begin by expressing $\rank(\mathbf{H}_i\mathbf{F})$ ($i\in\{c,s,e\}$) using the direct sum structure in Theorem~\ref{thm:mems_decomp}. 
Since $\{\CMcal{V}_j\}_{j\in\mathcal{K}}$ forms a direct sum decomposition of $\mathbb{C}^{n_t}$, there exist full-column-rank basis matrices $\mathbf{U}_j\in\mathbb{C}^{n_t\times k_j}$ with $\CMcal{C}(\mathbf{U}_j)=\CMcal{V}_j$ such that the concatenation
\begin{align}
    \mathbf{U} \triangleq [\mathbf{U}_n, \mathbf{U}_c, \mathbf{U}_s, \mathbf{U}_e, \mathbf{U}_{cs}, \mathbf{U}_{ce}, \mathbf{U}_{se}, \mathbf{U}_{cse}]
\end{align}
is nonsingular.
Hence, any precoder can be written uniquely as $\mathbf{F} = \sum_{j\in\mathcal{K}} \mathbf{U}_j \mathbf{G}_j$ for some coefficient blocks $\mathbf{G}_j\in\mathbb{C}^{k_j\times N_s}$. Define $\mathbf{F}_j \triangleq \mathbf{U}_j\mathbf{G}_j$, so that $\CMcal{C}(\mathbf{F}_j)\subseteq \CMcal{V}_j$ and $\rank(\mathbf{F}_j)=\rank(\mathbf{G}_j)$.

Next, based on the subspace definitions in Table~\ref{tab:subspace_dof} and their direct sum structure, the entire vector space can be partitioned into the null space $\CMcal{N}_i$ and the row space $\CMcal{R}_i$ for each channel $i \in \{c,s,e\}$. Consequently, the only components that lie in the respective row spaces are:
\begin{itemize}
    \item $\mathbf{H}_c$ sees $\{\CMcal{V}_c, \CMcal{V}_{cs}, \CMcal{V}_{ce}, \CMcal{V}_{cse}\}$,
    \item $\mathbf{H}_s$ sees $\{\CMcal{V}_s, \CMcal{V}_{cs}, \CMcal{V}_{se}, \CMcal{V}_{cse}\}$,
    \item $\mathbf{H}_e$ sees $\{\CMcal{V}_e, \CMcal{V}_{ce}, \CMcal{V}_{se}, \CMcal{V}_{cse}\}$.
\end{itemize}
Moreover, the restriction of $\mathbf{H}_c$ to $\CMcal{V}_c\oplus \CMcal{V}_{cs}\oplus \CMcal{V}_{ce}\oplus \CMcal{V}_{cse}$ is injective, so the rank of each effective channel output is equivalent to the rank of the vertically stacked coefficient blocks visible to that receiver.

Using this rank equivalence and the subadditivity of rank, we obtain
\begin{align}
\rank(\mathbf{H}_c\mathbf{F})
&\le \rank\left(\begin{bmatrix}\mathbf{G}_c \\ \mathbf{G}_{cs}\end{bmatrix}\right) + \rank\left(\begin{bmatrix}\mathbf{G}_{ce} \\ \mathbf{G}_{cse}\end{bmatrix}\right), \label{eq:rank_Hc_bound} \\
\rank(\mathbf{H}_s\mathbf{F})
&\le \rank\left(\begin{bmatrix}\mathbf{G}_s \\ \mathbf{G}_{cs}\end{bmatrix}\right) + \rank\left(\begin{bmatrix}\mathbf{G}_{se} \\ \mathbf{G}_{cse}\end{bmatrix}\right). \label{eq:rank_Hs_bound}
\end{align}
For the eavesdropper, adding blocks does not decrease the rank, so
\begin{align} \label{eq:rank_He_lower}
\rank(\mathbf{H}_e\mathbf{F})
=\rank\left(\begin{bmatrix}\mathbf{G}_e \\ \mathbf{G}_{ce} \\ \mathbf{G}_{se} \\ \mathbf{G}_{cse}\end{bmatrix}\right)
\ge \rank\left(\begin{bmatrix}\mathbf{G}_{ce} \\ \mathbf{G}_{se} \\ \mathbf{G}_{cse}\end{bmatrix}\right).
\end{align}

Substituting \eqref{eq:rank_Hc_bound}--\eqref{eq:rank_He_lower} into the objective and applying subadditivity $\rank([\mathbf{G}_i; \mathbf{G}_{cs}]) \le \rank(\mathbf{G}_i) + \rank(\mathbf{G}_{cs})$ yields
\begin{align} \label{eq:weighted_rank_sub}
&w_c\rank(\mathbf{H}_c\mathbf{F})-w_c\rank(\mathbf{H}_e\mathbf{F})+w_s\rank(\mathbf{H}_s\mathbf{F}) \nonumber \\
&\le w_c\rank(\mathbf{G}_c)+w_s\rank(\mathbf{G}_s)+(w_c+w_s)\rank(\mathbf{G}_{cs}) + \Gamma,
\end{align}
where $\Gamma$ captures the interaction between the shared blocks:
\begin{align}
\Gamma \hspace{-0.1em} \triangleq \hspace{-0.1em}&~ w_c\rank\left(\begin{bmatrix}\mathbf{G}_{ce} \\ \mathbf{G}_{cse}\end{bmatrix}\right) \hspace{-0.1em}+\hspace{-0.1em} w_s\rank\left(\begin{bmatrix}\mathbf{G}_{se} \\ \mathbf{G}_{cse}\end{bmatrix}\right) \hspace{-0.1em}-\hspace{-0.1em} w_c\rank\left(\begin{bmatrix}\mathbf{G}_{ce} \\ \mathbf{G}_{se} \\ \mathbf{G}_{cse}\end{bmatrix}\right).
\end{align}

We bound $\Gamma$ by considering two cases based on the weights $w_c$ and $w_s$. 

\medskip
\noindent\textit{Case 1: $w_s \ge w_c$.}
Using the lower bound $\rank([\mathbf{G}_{ce}; \mathbf{G}_{se}; \mathbf{G}_{cse}]) \ge \rank([\mathbf{G}_{se}; \mathbf{G}_{cse}])$, we have
\begin{align}
\Gamma &\le w_c\rank\left(\begin{bmatrix}\mathbf{G}_{ce} \\ \mathbf{G}_{cse}\end{bmatrix}\right) + (w_s-w_c)\rank\left(\begin{bmatrix}\mathbf{G}_{se} \\ \mathbf{G}_{cse}\end{bmatrix}\right).
\end{align}
Since $w_s - w_c \ge 0$, we can apply subadditivity $\rank([\mathbf{A}; \mathbf{B}]) \le \rank(\mathbf{A}) + \rank(\mathbf{B})$ to both terms to obtain an upper bound:
\begin{align}
\Gamma \le w_c\rank(\mathbf{G}_{ce}) + (w_s-w_c)\rank(\mathbf{G}_{se}) + w_s\rank(\mathbf{G}_{cse}). \label{eq:gamma_case1_a}
\end{align}
Alternatively, using the other lower bound $\rank([\mathbf{G}_{ce}; \mathbf{G}_{se}; \mathbf{G}_{cse}]) \ge \rank([\mathbf{G}_{ce}; \mathbf{G}_{cse}])$, the $-w_c$ terms cancel out, leaving
\begin{align}
\Gamma &\le w_s\rank\left(\begin{bmatrix}\mathbf{G}_{se} \\ \mathbf{G}_{cse}\end{bmatrix}\right) \le w_s\rank(\mathbf{G}_{se}) + w_s\rank(\mathbf{G}_{cse}). \label{eq:gamma_case1_b}
\end{align}
Combining \eqref{eq:gamma_case1_a} and \eqref{eq:gamma_case1_b} implies that $\Gamma$ is bounded by their minimum:
\begin{align}
\Gamma &\le \min\{ w_c\rank(\mathbf{G}_{ce}) + (w_s-w_c)\rank(\mathbf{G}_{se}), w_s\rank(\mathbf{G}_{se}) \} \nonumber \\
&\quad + w_s\rank(\mathbf{G}_{cse}) \nonumber \\
&= w_s\rank(\mathbf{G}_{se}) - w_c[\rank(\mathbf{G}_{se})-\rank(\mathbf{G}_{ce})]^+ \nonumber \\
&\quad + w_s\rank(\mathbf{G}_{cse}). \label{eq:gamma_case1_final}
\end{align}

\medskip
\noindent\textit{Case 2: $w_c > w_s$.}
To avoid negative coefficients, we split the penalty term as $-w_c(\cdot) = -w_s(\cdot) - (w_c-w_s)(\cdot)$. Applying $\ge \rank([\mathbf{G}_{se}; \mathbf{G}_{cse}])$ to the first part and $\ge \rank([\mathbf{G}_{ce}; \mathbf{G}_{cse}])$ to the second part, we obtain
\begin{align}
\Gamma \le &~ w_c\rank\left(\begin{bmatrix}\mathbf{G}_{ce} \\ \mathbf{G}_{cse}\end{bmatrix}\right) + w_s\rank\left(\begin{bmatrix}\mathbf{G}_{se} \\ \mathbf{G}_{cse}\end{bmatrix}\right) \nonumber \\
&- w_s\rank\left(\begin{bmatrix}\mathbf{G}_{se} \\ \mathbf{G}_{cse}\end{bmatrix}\right) - (w_c-w_s)\rank\left(\begin{bmatrix}\mathbf{G}_{ce} \\ \mathbf{G}_{cse}\end{bmatrix}\right) \nonumber \\
=&~ w_s\rank\left(\begin{bmatrix}\mathbf{G}_{ce} \\ \mathbf{G}_{cse}\end{bmatrix}\right) \le w_s\rank(\mathbf{G}_{ce}) + w_s\rank(\mathbf{G}_{cse}). \label{eq:gamma_case2_a}
\end{align}
Note that the bound \eqref{eq:gamma_case1_b} remains universally valid for all positive weights since we only dropped the non-positive $-w_c\rank([\mathbf{G}_{ce}; \mathbf{G}_{cse}])$ term from a valid upper bound. Therefore, for $w_c > w_s$, combining \eqref{eq:gamma_case1_b} and \eqref{eq:gamma_case2_a} gives
\begin{align}
\Gamma &\le \min\{ w_s\rank(\mathbf{G}_{ce}), w_s\rank(\mathbf{G}_{se}) \} + w_s\rank(\mathbf{G}_{cse}) \nonumber \\
&= w_s\rank(\mathbf{G}_{se}) - w_s[\rank(\mathbf{G}_{se})-\rank(\mathbf{G}_{ce})]^+ \nonumber \\
&\quad + w_s\rank(\mathbf{G}_{cse}). \label{eq:gamma_case2_final}
\end{align}

Combining \eqref{eq:gamma_case1_final} and \eqref{eq:gamma_case2_final} yields the unified expression for any $w_c, w_s > 0$:
\begin{align}
\Gamma &\le w_s\rank(\mathbf{G}_{se}) - \min\{w_c, w_s\}[\rank(\mathbf{G}_{se})-\rank(\mathbf{G}_{ce})]^+ \nonumber \\
&\quad + w_s\rank(\mathbf{G}_{cse}). \label{eq:gamma_final}
\end{align}

Substituting \eqref{eq:gamma_final} into \eqref{eq:weighted_rank_sub} and using $\rank(\mathbf{G}_j)=\rank(\mathbf{F}_j)$ proves \eqref{prop:weighted_inequality}. Finally, note that the right-hand side of \eqref{prop:weighted_inequality} is monotonically non-decreasing with respect to every $\rank(\mathbf{F}_j)$. Specifically, for $x = \rank(\mathbf{F}_{se})$ and $y = \rank(\mathbf{F}_{ce})$, the function $w_s x - \min\{w_c,w_s\}[x-y]^+$ is non-decreasing in both $x$ and $y$, because its slope with respect to $x$ is at least $w_s - \min\{w_c,w_s\} \ge 0$. Therefore, together with Lemma~\ref{lem:2}, substituting the maximum spatial dimensions $\rank(\mathbf{F}_j) \le \dim(\CMcal{V}_j) = k_j$ yields the valid upper bound $d(\mathbf{F}) \le d_{\text{max}}$.
\end{IEEEproof}

\bibliographystyle{IEEEtran}
\bibliography{ref_MEMS}

\end{document}